\documentclass[superscriptaddress,pra,aps,twocolumn,10pt,longbibliography]{revtex4-1}
\usepackage{amsfonts}
\pdfoutput=1
\usepackage{amssymb}
\usepackage{amsmath}
\usepackage{amsthm}
\usepackage{graphicx}
\usepackage{verbatim}
\usepackage{hyperref}
\usepackage{bbm}
\usepackage[english]{babel}
\usepackage[normalem]{ulem}
\usepackage{natbib}
\usepackage[pdftex]{color}
\usepackage[dvipsnames]{xcolor}
\usepackage{braket}
\usepackage[utf8]{inputenc}
\usepackage{enumitem}
\usepackage{subcaption}

\newcommand{\ep}{\mathcal{E}}
\newcommand{\dualep}{\ep^\#}
\newcommand{\tr}{\mathrm{Tr}}

\newcommand{\qaoa}{\mathrm{QAOA}}

\newcommand{\ketbra}[2]{|#1\rangle \langle #2 |}

\newtheorem{definition}{Definition}

\hypersetup{
    colorlinks=true,       
    linkcolor=cyan,          
    citecolor=magenta,        
    filecolor=magenta,      
    urlcolor=cyan,           
    runcolor=cyan
}

\graphicspath{{figs/}}

\begin{document}

\title{Dual Map Framework for Noise Characterization of Quantum Computers}

 \author{James Sud}
 \email{jsud@usra.edu}
 \affiliation{Quantum Artificial Intelligence Laboratory (QuAIL), NASA Ames Research Center, Moffett Field, CA, 94035, USA}
 \affiliation{USRA Research Institute for Advanced Computer Science (RIACS), Mountain View, CA, 94043, USA}

 \author{Jeffrey Marshall}
 \affiliation{Quantum Artificial Intelligence Laboratory (QuAIL), NASA Ames Research Center, Moffett Field, CA, 94035, USA}
 \affiliation{USRA Research Institute for Advanced Computer Science (RIACS), Mountain View, CA, 94043, USA}

 \author{Zhihui Wang}
 \affiliation{Quantum Artificial Intelligence Laboratory (QuAIL), NASA Ames Research Center, Moffett Field, CA, 94035, USA}
 \affiliation{USRA Research Institute for Advanced Computer Science (RIACS), Mountain View, CA, 94043, USA}

 \author{Eleanor Rieffel}
 \affiliation{Quantum Artificial Intelligence Laboratory (QuAIL), NASA Ames Research Center, Moffett Field, CA, 94035, USA}
 
  \author{Filip A. Wudarski}
 \email{fwudarski@usra.edu}
 \affiliation{Quantum Artificial Intelligence Laboratory (QuAIL), NASA Ames Research Center, Moffett Field, CA, 94035, USA}
 \affiliation{USRA Research Institute for Advanced Computer Science (RIACS), Mountain View, CA, 94043, USA}

\date{\today}

\begin{abstract}

In order to understand the capabilities and limitations of quantum computers, it is necessary to develop methods that efficiently characterize and benchmark error channels present on these devices. In this paper, we present a method that faithfully reconstructs a marginal (local) approximation of the effective noise (MATEN) channel, that acts as a single layer at the end of the circuit. We first introduce a dual map framework that allows us to analytically derive expectation values of observables with respect to noisy circuits. These findings are supported by numerical simulations of the quantum approximate optimization algorithm (QAOA) that also justify the MATEN, even in the presence of non-local errors that occur during a circuit. Finally, we demonstrate the performance of the method on Rigetti’s Aspen-9 quantum computer for QAOA circuits up to six qubits, successfully predicting the observed measurements on a majority of the qubits. 

\end{abstract}

\maketitle

\section{Introduction}\label{introduction}

Appropriate and accurate error characterization and benchmarking is vital for many aspects of quantum computation. Understanding dominant forms of error allows for improvements on quantum hardware, bringing these devices closer to the fault-tolerant regime, and possibly allowing for the tailoring of error correcting codes to specific error channels \cite{err-mitt-EC}. On the algorithms side, error characterization opens the possibility for error-aware algorithm design and error mitigation strategies, improving the performance of algorithms on hardware \cite{zne1, zne2}. A plethora of protocols have been designed for understanding error. These can be divided into benchmarking protocols, which aim to return numerical values that capture the rate of errors in a process (usually defined as an average fidelity \cite{nielsen02,wudarski20}), and characterization protocols, which aim to return information about both the level and form of the error channels themselves. Benchmarking protocols include randomized benchmarking \cite{emerson05, magesan11} (along with extensions such as \cite{claes21}), cycle benchmarking \cite{erhard19}, and direct fidelity estimation \cite{flammia11}. Characterization protocols include quantum process tomography \cite{chuang97}, gate set tomography \cite{greenbaum15}, Hamiltonian estimation \cite{schirmer04}, and robust phase estimation \cite{kimmel15}, as well as state preparation and measurement (SPAM) error characterization methods such as \cite{sun20, lin21, werninghaus21}. So far, benchmarking and characterization methods have suffered substantial shortcomings - either returning limited information (e.g. average fidelity for RB) or restricted to small systems due to exponential scaling (tomographic methods). In this work, we develop a characterization scheme that efficiently returns information about process matrix of the marginal noise channel acting on a single qubit. The method combines ease and efficiency of benchmarking techniques with substantially richer information content. Additionally, the introduced protocol operates without additional compilation overhead, as opposed to RB approaches, which require twirling subroutine to cast the noisy channel into a convenient form of a Pauli channel.

Quantum noise, which can lead to computational errors, are inevitable companions of quantum evolution. In order to properly describe physically admissible errors, one has to employ the framework of completely positive and trace preserving (CPTP) maps, which are referred to as error channels. These channels can be represented in numerous ways \cite{milz17,breuer02,bengtsson17}. For our purposes, the most natural representation is of the following form
\begin{equation}\label{eq:chi_map}
\ep[\rho] = \sum_{k,l=0}^{d^2-1}\chi_{k,l}P_k\rho P_l^\dag,
\end{equation}
where $P_k$ are operators, $d=2^N$ is dimensionality of the Hilbert space for $N$ qubits and $\chi$ is referred to as a process matrix. Setting $P_k$ to orthonormal basis elements (e.g. Pauli matrices), one can determine all elements of $\chi$ matrix via quantum process tomography \cite{chuang97}. In order to represent a valid quantum channel, Eq.~\eqref{eq:chi_map} has to be CPTP, which happens when $\chi\ge0$ (CP condition) and 
\begin{equation}
    \sum_{k,l=0}^{n^2-1}\chi_{k,l}P_l^\dag P_k = \mathbbm{1},\quad\quad \mathrm{(TP\ condition)}. 
\end{equation}

On noisy intermediate-scale quantum (NISQ) devices, levels of noise are too high for error correction and fault tolerance to occur. Thus, error mitigation, and error-aware algorithm co-design strategies are needed to maximize the performance of algorithms run on these devices. In order to determine these optimal mitigation and co-design strategies, it is imperative to characterize and understand error. A popular and well studied algorithm for NISQ devices is the Quantum Approximate Optimization Algorithm or Quantum Alternating Operator Ansatz (QAOA) \cite{farhi14, hogg00, hadfield19}, which aims to find approximately optimal solutions to optimization problems. In this work, we study the application of our characterization method for QAOA run on combinatorial optimization problems. The characterization and effect of local noise in QAOA circuits has been previously studied \cite{marshall20, xue21, wang2021noiseinduced,streif2021}, however here we provide analytical treatment for popular classes of error channels, as well as for a generic single-qubit noise channel. 

Given this understanding of prior work on error characterization and benchmarking, as well as the introduction of error maps and QAOA, we lay out the rest of our paper as follows. In Sec.~\ref{dualmapframework}, we introduce the framework of the dual map that is essential for analytical derivation of the expectation values of arbitrary observables in the noisy setup.
In Sec.~\ref{singlequbitnoisecharacterization}, we use these results to reverse-engineer a method to introduce a marginal approximation to the effective noise (MATEN), the core procedure of this contribution, that allows to estimate local contribution to the error channels. In Sec.~\ref{localvsnonlocalchannels}, we discuss limitations of the MATEN protocol for spatially correlated noise. Finally, in \ref{results} we demonstrate the efficacy of the method in characterizing error on classical simulations, as well as on the Aspen-9 quantum computer device from Rigetti Computing. 

\section{Dual Map Framework}\label{dualmapframework}

Given a quantum state $\rho$, an error channel $\ep$ (in the form of Eq.~\eqref{eq:chi_map}), and an Hermitian operator $O$, the noisy expectation value of $O$ with respect to $\rho$ is typically evaluated as $\mathrm{Tr} [ O \ep(\rho)]$. In this work, however, we consider the dual action of $\ep$, and compute the same expectation as 
\begin{equation}\label{eq:dual_map}
   \langle O \rangle = \mathrm{Tr} [\dualep(O) \rho] = \mathrm{Tr} [O' \rho]
\end{equation}
where $\dualep$ is the dual channel of $\ep$, defined as 
\begin{equation}
    \ep^\#[O] = \sum_{k,l=0}^{n^2-1}\chi_{k,l}P_l^\dag O P_k,
\end{equation}
which can be derived from the cyclic property of trace.

We can then study properties of the modified operator $O' = \ep^\# (O)$, which means that noise affects only the observable and not the state $\rho$. Therefore, the expectation values with respect to the ideal quantum state $\rho$ (e.g. output of a quantum circuit), can potentially benefit from the local structure of noise of the observables. In particular, if $\rho$ is a pure state (i.e. $\rho^2=\rho$), we can avoid costly simulations of the density matrix, and focus on unitary simulations and local measurements of a state vector.

Finally note that it is always possible to decompose a quantum map $\Lambda$ associated with a noisy quantum circuit as a composition $\Lambda = \mathcal{E}\circ \mathcal{U}$, where $\mathcal{U}$ is the circuit's (ideal) unitary map, and $\mathcal{E}$ a noisy channel (e.g with the trivial example $\ep = \Lambda\circ\mathcal{U}^{\dag}$). Therefore, instead of characterizing the total error map $\Lambda$ (which includes contribution from the ideal unitary), we focus on determining $\ep$, which one can perceive as an effective noise channel for the circuit (in general dependent on $\mathcal{U}$, e.g. rotational angles in QAOA). This mathematical trick, allows us to ``move'' effects of noise to the very last layer of the quantum circuit (see Fig.~\ref{fig:noise_circuit_diagram}), and exploit the dual map framework. The main advantage of using this formalism, is that many observables of interest (e.g. combinatorial or molecular Hamiltonians) can be expressed as a combination of $k$-local terms, and, as explained above, the simulation of which can be significantly more efficient in the dual map framework. We demonstrate this idea in the subsequent examples.

\begin{figure*}
    \centering
    \includegraphics[width=0.98\textwidth]{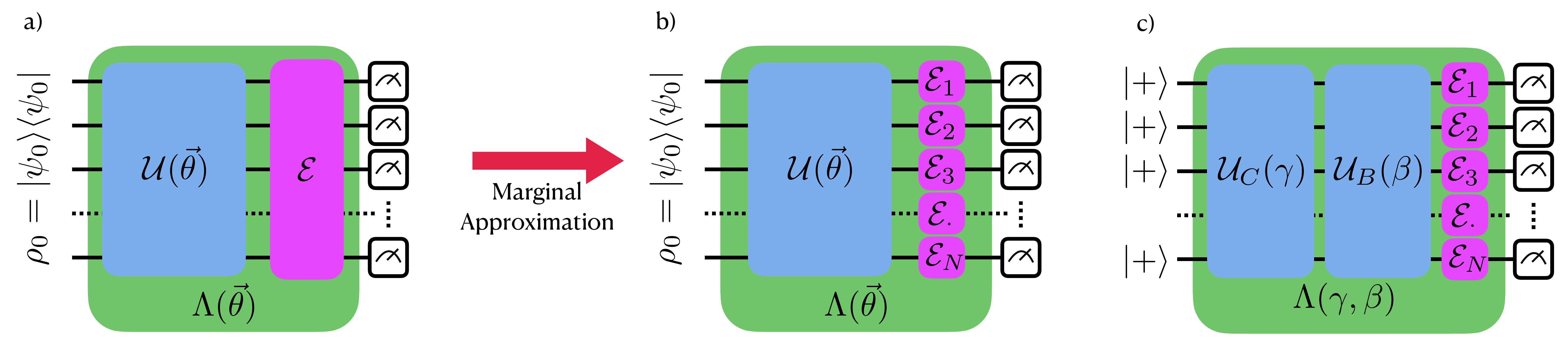}
    \caption{Operational framework for the noisy circuit characterization that is described by a quantum map $\Lambda(\vec{\theta})$, where $\vec{\theta}$ is a collection of circuit parameters. In a)
    noise is ``moved'' to the last layer (according to $\ep = \Lambda(\vec{\theta})\circ\mathcal{U}^\dag(\vec{\theta})$, and in general it is of non-local character. The non-locality is neglected in b) by approximating circuit with only local quantum channels $\ep_i$ acting on $i$-th qubit. In c) we depict the marginal approximation for a single-layered QAOA. 
    }
    \label{fig:noise_circuit_diagram}
\end{figure*}

\subsection{Example: Noisy Single-Layered QAOA}\label{noisysinglelayeredqaoa}

In this section we analyze the example of QAOA circuits, which are constructed by interleaving layers of parameterized unitaries of a mixing Hamiltonian $B$ and a phasing Hamiltonian $H$, as so

\begin{equation}\label{eq:qaoa_ansatz}
    \ket{\Psi_\qaoa(\vec{\gamma}, \vec{\beta})} = e^{-i\beta_pB}e^{-i\gamma_pH}\cdots e^{-i\beta_1B}e^{-i\gamma_1H} \ket{\psi_{\mathrm0}},
\end{equation}
where $p$ is the number of layers in the circuit, $(\vec{\gamma}, \vec{\beta})$ represent length $p$ parameter vectors, and $\ket{\psi_{\mathrm0}}$ corresponds to an initial state. Given this form, one can then choose $(\vec{\gamma}, \vec{\beta})$ such that the expectation value of $H$ is optimized (minimized or maximized) when the circuit is applied to a suitably chosen initial state. Strategies for optimizing $(\vec{\gamma}, \vec{\beta})$ \cite{streif20, zhou20, shaydulin19, brandao18} as well as choosing optimal starting states and mixing Hamiltonians \cite{hadfield19, sack21, wang20} have been intensively analyzed. In the original formulation, and most applications of QAOA, the cost function is classical, ensuring that its corresponding Hamiltonian  consists of Pauli terms only containing the Pauli $Z$ operators. For this section, we restrict to QAOA applied to the well-studied form of quadratic unconstrained binary optimization (QUBO), with cost functions given by a Hamiltonian of the form \begin{equation}\label{eq:localQAOAop}
    H = H_1 + H_2 = \sum_i h_i Z_i + \sum_{ij}J_{ij} Z_i Z_j.
\end{equation}
Many popular combinatorial optimization problems can be cast into QUBO form \cite{glover19}. In order to understand the effects of various noise channels on specific problems under certain noise assumptions then, it suffices to compute the action of the dual channel on Pauli terms with limited locality, and analyze how the modified $H'$ Hamiltonians relate to the original cost Hamiltonians. 

For parameterized circuits, such as QAOA, in order to perform mathematical analysis, we assume $\mathcal{E}$ is independent of the parameters $(\vec{\gamma}, \vec{\beta})$, and is a product of local channels such that we can write
\begin{equation}
    \rho(\vec{\gamma}, \vec{\beta}) = \mathcal{E} \mathcal{U}(\vec{\gamma}, \vec{\beta}) \left[ \ketbra{\psi_0}{\psi_0} \right].
\end{equation}
This assumption for QAOA circuits is visualized in Fig.~\ref{fig:noise_circuit_diagram}~c), with $\ep = \bigotimes_{i=1}^N \ep_i$ that is equivalent to MATEN (thoroughly described in the next section). A similar noise structure was considered in \cite{marshall20,xue21}. 

In the following subsections, we analyze the effects of the dual map of various common error channels on Hamiltonians of form Eq.~\eqref{eq:localQAOAop} wit one layer. For these examples, we assume the error channel is identical on each qubit in order to simplify the equations, but this assumption can straightforwardly be relaxed.

\subsubsection{Single Qubit Depolarizing Channel}\label{singlequbitdepolarizingchannel}

We first define a single qubit depolarizing channel, parameterized by a depolarization rate $p$, as
\begin{equation}\label{eq:depolarizing}
    \ep_p^i(\rho) = \frac{1+3p}{4}\rho+\frac{1-p}{4}\sum_{k=1}^3\sigma_k^i\rho \sigma_k^i,
\end{equation}
with $\sigma_1, \sigma_2, \sigma_3$ corresponding to Pauli X,Y, and Z respectively, and the indices $i$ corresponding to the qubit that Pauli operators act upon. Note that each Pauli matrix is an eigenmatrix of the depolarizing channel with eigenvalue $p$, i.e. $\ep_p(\sigma_k)=p\sigma_k$, and this map is self-dual ($\ep=\dualep$), so we also have $\ep^\#_p(\sigma_k)=p\sigma_k$.

For an $N$ qubit system we have noise channel acting on each qubit, i.e. $\ep_p = \ep_p^1\otimes\ep_p^2\otimes\ldots\otimes \ep_p^n$, where $\ep_p^i$ corresponds to a one-local channel on qubit $i$. Therefore, if $\ep_p$ acts on a $k$-local term in Hamiltonian, and we assume that $p$ is constant on all qubits, $\ep_p$ effectively multiplies this term by $p^k$. Specifically, we have
\begin{gather}
    \dualep_p(Z_i) = p Z_i, \\ 
    \dualep_p(Z_i Z_j) = p^2 Z_i Z_j.
\end{gather}
This allows us to easily identify the action of local depolarizing noise on QAOA for depth one with the noise channel applied at the end of the circuit, by moving to the dual picture
\begin{align}
    H' = \ep_p^\#(H) = p H_1 + p^2 H_2.
\end{align}
Thus, for single qubit depolarizing channels, the effect on QAOA cost operators is simply that one-qubit terms are rescaled by $p$, two-qubit terms are rescaled by $p^2$, and $k$-qubit terms by $p^k$ (although $k>2$ are not considered for QUBO problems). For a strictly 2-local problem such as MaxCut, this would mean that the cost is simply rescaled by $p^2$. For optimization purposes, this simple rescaling means that the optimal parameter settings stay unchanged.

\subsubsection{Amplitude Damping}\label{amplitudedamping}

Another common error channel is amplitude damping, given by the following map 
\begin{equation}\label{eq:amp_damp_map}
    \ep{\rho} = A_1 \rho A_1^\dag + A_2 \rho A_2^\dag,
\end{equation}
where $A_1$, $A_2$ are Kraus operators parameterized by a damping rate $\gamma$ and given by 
\begin{equation}
    A_1 = \left(\begin{array}{cc}1 & 0 \\0 & \sqrt{1-\gamma}\end{array}\right)\quad,\quad A_2 = \left(\begin{array}{cc}0 & \sqrt{\gamma} \\0 & 0\end{array}\right).
\end{equation}
The action of the dual of this error channel on single and two qubit Pauli Z operators are as follows
\begin{gather}
    \ep_{\gamma}^\#(Z_i) = (1-\gamma)Z_i + \gamma I, \\ 
    \ep_{\gamma}^\#(Z_iZ_j) = (1-\gamma)^2Z_iZ_j + \gamma(1-\gamma)(Z_i+Z_j) + \gamma^2I.
\end{gather}
We can see then that the 1-local terms are simply scaled and shifted. For the 2-local terms, we get not only a scale and a shift from the first and third terms, respectively, but also an extra contribution of 1-local terms from the middle term.
We can write out the action of this channel on the general Hamiltonian given in Eq.~\eqref{eq:localQAOAop}
\begin{align}
    & H' =(1-\gamma) H_1 + \gamma \sum_i h_i + (1-\gamma)^2 H_2 \nonumber \\
    &+ \gamma^2 \sum_{i<j}J_{ij}
    +\gamma(1-\gamma)\sum_{i<j}J_{ij}(Z_i+Z_j).
\end{align}
Now the only term that is neither a scale nor a constant shift is the last term. We first note that if $h_i=0$ for all $i$ and we start in a $\mathbb{Z}_2$ symmetric state, the resultant QAOA state is $\mathbb{Z}_2$ symmetric, thus all single qubit $Z$ terms go to  zero \cite{shaydulin21}, so this added 1-local term has no effect on the observed cost function value. 

Another case where this term has a nice solution can be seen as follows: 
We can rewrite $\sum_{i<j}J_{ij}(Z_i+Z_j)$ as $\sum_i Z_i (\sum_{j \neq i} J_{ij})$.
Next, if $(\sum_{j \neq i} J_{ij}) = ah_i$ for all $i$ and for some constant $a$, then this term gives us $a H_1$, meaning that $H_1$ is rescaled by $1-\gamma+a$ instead. This occurs in some cases enumerated below
\begin{enumerate}
    \item 
    if all $h$'s and $J$'s are constant (all equal $h$, $J$, respectively):
    \begin{equation}
        \sum_i Z_i (\sum_{j \neq i} J_{ij}) = \frac{J (N-1)}{h} H_1,
    \end{equation}
    \item 
    $d$-regular graph, all $h$'s constant, all nonzero $J$'s are constant (all equal $h$, $J$, respectively):
    \begin{equation}
        \sum_i Z_i (\sum_{j \neq i} J_{ij}) = \frac{J d}{h} H_1,
    \end{equation}
    \item 
    max-$k$-colorable-subgraph \cite{wang20} ($h_i=d_i$ where $d_i$ is degree of vertex $i$, $J_{ij}=-1$ if edge $(i,j)$ exists in the graph):
    \begin{equation}
        \sum_i Z_i (\sum_{j \neq i} J_{ij}) = -\sum_i d_i Z_i = -H_1.
    \end{equation}
\end{enumerate}
Notably, if $a=-1$ as in case 3), the Hamiltonian reduces to 
\begin{equation}
    H' = (1-\gamma)^2 H + \gamma \sum_i h_i + \gamma^2 \sum_{i<j}J_{ij},
\end{equation}
where we see that the entire Hamiltonian is simply scaled and shifted.

For an analysis of the effects of other common error channels on QAOA operators, such as Pauli Channels, T1/T2 error, and overrotations, etc), please see Appendix \ref{app:error_maps}.

\section{Single Qubit Noise Characterization}\label{singlequbitnoisecharacterization}
So far we have shown how certain local noise channels affect 1-local and 2-local observables, given complete knowledge of the noise. In this section, however, we demonstrate the opposite direction, showing how to exploit the dual map framework to find a marginal approximation to the effective noise (MATEN), which is defined as follows
\begin{definition}[MATEN]
For a unitary quantum circuit $\mathcal{U}$ acting on $N$-qubits, and its noisy realization $\Lambda_U$, we call $\ep = \Lambda_U\circ\mathcal{U}^{\dag}$ an effective noise channel, that acts as the final CPTP circuit layer. Additionally we define a marginal approximation to the effective noise (MATEN) as
\begin{equation}
    \tilde{\ep} = \bigotimes_{k=1}^N \tr_{\bar{k}}\left(\ep\right) = \bigotimes_{k=1}^N \ep_k,
\end{equation}
where $\tr_{\bar{k}}(\ep) = \ep_k$ traces out all subsystem except $k$-th (see~Fig.~\ref{fig:noise_circuit_diagram} b)).
\end{definition}

In order to determine a MATEN, we express a noisy map in terms of so-called process matrix (or $\chi$ matrix), that can be in principle measured directly in a set of experiments via quantum process tomography \cite{chuang97}. The map takes the form (for a single qubit),
\begin{equation}
    \mathcal{E}_\chi(\rho) = \sum_{k,l=0}^3 \chi_{kl}\sigma_k \rho \sigma_l,
\end{equation}
where $\chi_{kl}$ are elements of the $\chi$ matrix, which in general can be expressed as
\begin{equation}\label{eq:chi_matrix}
    \chi = \left(
\begin{array}{cccc}
 p_0 & t_{0,1}+i v_{0,1} & t_{0,2}+i v_{0,2} & t_{0,3}+i v_{0,3} \\
 t_{0,1}-i v_{0,1} & p_1 & t_{1,2}-i t_{0,3} & t_{1,3}+i t_{0,2} \\
 t_{0,2}-i v_{0,2} & t_{1,2}+i t_{0,3} & p_2 & t_{2,3}-i t_{0,1} \\
 t_{0,3}-i v_{0,3} & t_{1,3}-i t_{0,2} & t_{2,3}+i t_{0,1} & p_3 \\
\end{array}
\right),
\end{equation}
with $t_{kl}$ and $v_{kl}$ representing real and imaginary parts of $\chi_{kl}$ elements, respectively. This form, combined with conditions $\chi\ge0$ and $\sum_{k=0}^3p_k=1$, guarantees that the map is completely positive and trace preserving, and is the most general for the qubit systems. Note, that in total we have 12 free parameters, and diagonalizing the $\chi$ matrix will lead to a Kraus form (note that the Kraus form is not unique). Given $\chi$ we can then evaluate the effect of the dual map on the Pauli observables 

\begin{align}
    \tilde{I} &=  I ,\label{eq:noisyi}\\
    \tilde{X} &=  (p_0+p_1-p_2-p_3)X + 4t_{01}I \nonumber \\ &+2((t_{12}-v_{03})Y+(t_{13}+v_{02}))Z, \label{eq:noisyx}\\
    \tilde{Y} &=  (p_0+p_2-p_1-p_3)Y + 4t_{02}I \nonumber \\ &+2((t_{23}-v_{01})Z+(t_{12}+v_{03}))X, \label{eq:noisyy}\\
    \tilde{Z} &=  (p_0+p_3-p_1-p_2)Z + 4t_{03}I \nonumber \\ &+2((t_{13}-v_{02})X+(t_{23}+v_{01}))Y,  \label{eq:noisyz}
\end{align}
where $\tilde{I}$,$\tilde{X}$,$\tilde{Y}$,$\tilde{Z}$ represent the noisy transformations of the Pauli operators, and $\tilde{I}=I$ due to the property that the dual of trace preserving maps are unital (i.e. $\dualep(I)=I$). We can rewrite the coefficients in each equations as $P_{AB}$, forming a vector of coefficients $\vec{P}$, so for example 
$\tilde{X} = P_{XI}I+P_{XX}X+P_{XY}Y+P_{XZ}Z$. We can also write a simple matrix $A$ that relates coefficients $P_{AB}$ to the $\chi$ matrix elements as in $\vec{P} = A\vec{\chi}$, where $\vec{\chi}$ is a 12-dimensional vector having all independent $\chi$ matrix elements (i.e. $p_k, t_{kl}$, and $v_{kl}$).

Given Eqs.~\eqref{eq:noisyi}-\eqref{eq:noisyz} we can then perform the following procedure for a parameterized circuit of interest\footnote{Extension to parameter-free circuit is straightforward, and requires only altering some gates, e.g. $X\to Z$. However, this procedure would disturb the investigated algorithm, and could serve only as a characterization protocol.}:
\begin{enumerate}
    \item Choose a set $\mathbb{S}$ of parameters to the circuit. E.g. for level-1 QAOA this corresponds to choosing $|\mathbb{S}|$ different ($\gamma_1$, $\beta_1$) pairs. 
    \item Implement the circuit on a quantum device, take many measurements in the $X$, $Y$, and $Z$ bases to approximate $\braket{\tilde{X}}$, $\braket{\tilde{Y}}$, and $\braket{\tilde{Z}}$ for each parameter setting in $\mathbb{S}$ on each qubit. Since $\braket{\tilde{I}}$ is trivial, the measurement is not needed.
    \item On a classical simulator or via analytic derivation, determine the ideal values of the $\braket{X}$, $\braket{Y}$, and $\braket{Z}$ for each parameter setting in $\mathbb{S}$ for each qubit. $\braket{I}$ is trivial to calculate. \label{idealstep}
    \item Using the ideal and noisy values of all four Pauli observables, determine the coefficients $\vec{P}$ via linear regression on Eqs.~\eqref{eq:noisyi}-\eqref{eq:noisyz} for each qubit. \label{regressionstep}
    \item Given the coefficients $\vec{P}$, along with the matrix $A$ relating $\vec{P}$ to $\vec{\chi}$ matrix elements, perform $\vec{\chi}_{pred} = A^{-1}\vec{P}$ for each qubit, where $\vec{\chi}_{pred}$ are the elements of the predicted $\chi$ matrix.
\end{enumerate}
The above protocol is visualized in the chart in Fig.~\ref{fig:protocol}.
\begin{figure*}
    \centering
    \includegraphics[width=0.99\textwidth]{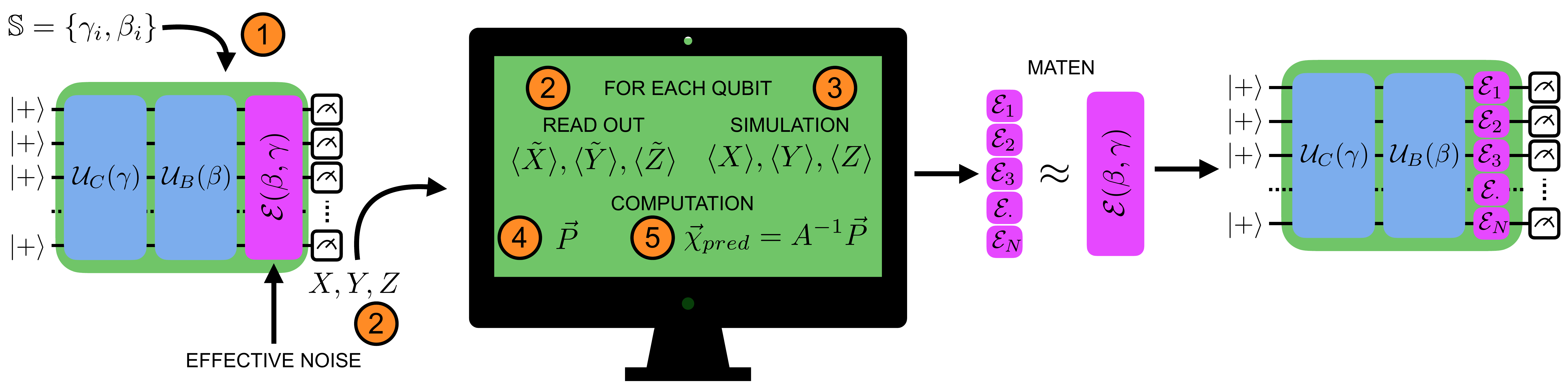}
    \caption{Protocol for characterizing local approximation to effective noise in the case of a single-layered QAOA algorithm. 1) run multiple times a circuit with different parameters from the set $\mathbb{S}$, next 2) measure all qubit registers in $X, Y$, and $Z$ bases. Use output bitstrings to infer (on a classical computer) noisy expectation values $\langle \tilde{X}\rangle, \langle \tilde{Y}\rangle$, and $\langle \tilde{Z}\rangle$. For the same set of parameters, 3) simulate ideal circuits on a classical computer to obtain expectation values.  Based on the ideal expectation values 4) construct vector $\vec{P}$, and then 5) determine $\vec{\chi}_{pred}$. Perform this procedure for each qubit in order to reconstruct the MATEN that approximates effective noise.}
    \label{fig:protocol}
\end{figure*}
We note two limitations with the presented procedure. First is that of step \ref{idealstep}, in general, it may be prohibitive 
to determine the ideal values of single qubit expectation values in simulation. However, for shallow circuits, one can use reverse light cone arguments to calculate local operator expectation values in time and memory growing exponentially with circuit depth, rather than circuit size \cite{peng20}. Additionally, if the noise channels \emph{mildly} depend on circuits \footnote{Here by mildly we mean, that noise is static, and parameter (e.g. angle) independent to the leading order.}, one could perform this characterization process on a few sets of qubits individually; this approach would work especially for shallow circuits. Finally, state of the art classical simulators can handle circuits with relatively large depth and qubit number, depending on simulation methods and computational resources. 

Second, for problems with $\mathbb{Z}_2$ symmetry, the ideal values of $\braket{Y}$ and $\braket{Z}$ vanish, so it may be impossible to fully determine $\vec{P}$, and it remains an open question if we can reliably determine nonzero elements of $\vec{P}$. If this is the case, one can derive similar equations as Eqs.~\eqref{eq:noisyi}-\eqref{eq:noisyz}, but for two-qubit operators, although this becomes much more complicated. For our analysis, we restrict to problems that lack $\mathbb{Z}_2$ symmetry. For a problem such as MaxCut, this can be achieved by simply adding single qubit $Z$ terms to the Hamiltonian. Presumably, these single qubit $Z$ gates do not introduce a significant amount of noise (on Rigetti devices, they are indeed implemented in software), so the $\chi_{pred}$ matrix should remain close to that of the original circuit. Thus, the characterized channels for these modified problems should match very well those of the original problems. One can also break this symmetry by starting in a different initial state. For QAOA problems the initial state is usually $\ket{+}^{\otimes N}$, which is $\mathbb{Z}_2$ itself, but changing $\ket{+}$ to a different non-$\mathbb{Z}_2$ symmetric state would break that symmetry.

\section{Local vs Non-Local Channels}\label{localvsnonlocalchannels}

One of the major challenges in current technology is understanding spatial correlations in noise. Whether or not noise is confined locally to a single qubit, or can be correlated across neighboring (or even distant) qubits (such as in crosstalk \cite{ash20, sarovar20}) determines the efficacy of error mitigation techniques, and quantum error correction (where errors are typically assumed to be independent). Here we aim to find out, how well one can approximate non-local noise channels with the MATEN approach. Our strategy is as follows: i) first we derive a lower bound for the worst case scenario, ii) then we numerically compute accuracy of the method for random non-local channels, iii) finally we repeat numerical analysis from ii), but for random Pauli channels and analyze some scaling properties. 

Since single-qubit $\chi$ matrix (in Pauli basis) is a positive operator of trace one, we can treat it as a 4-dimensional quantum state (with some extra constraints imposed by the structure of $\chi$). This enables us to incorporate results from the theory of quantum entanglement for the analysis of non-local channels. In particular, all the marginal states for maximally entangled states are maximally mixed states, i.e. they are proportional to the identity matrix. Therefore, the marginal approximation (MA), which on the level of $\chi$ matrix is translated to 
\begin{equation}
    \chi = \ket{\Psi}\bra{\Psi} \to \chi_{\mathrm{MA}} 
    =\bigotimes_{k=1}^N \left(\frac{1}{4}\mathbbm{1} \right) = \frac{1}{4^N}\mathbbm{1}_{4^N},
\end{equation}
also yields the maximally mixed state in the full $4^N$ dimensional space (where $N$ is the number of considered qubits), which corresponds to the fully depolarizing channel. Above we denote the non-local process matrix $\chi$ as the maximally entangled state, that is defined as a projector onto $\ket{\Psi}=\frac{1}{2}\sum_{i=0}^3\ket{ii\ldots i}$, with $N$ 4D subsystems (each corresponding to a qubit), note that this is a GHZ state \cite{greenberger2007going}. We conjecture that the effective channel with the maximally entangled $\chi$ matrix is the worst case scenario for the proposed MATEN protocol. 
Since the MATEN approach neglects all non-trivial correlations between different subsystems, and maximally entangled states exhibit the strongest correlations among quantum objects resulting in minimal knowledge of the subsystem's structure (maximally mixed state),  the protocol yields the minimum fidelity value between the marginal approximation (MA) and the full $\chi$. However, this conjecture requires more rigorous treatment, which we leave as an open problem. Note that maximally entangled $\chi$ is a completely valid choice, since $\chi\ge0$ and the map associated with it is trace preserving. 

Having established that the maximally entangled $\chi$ matrix is the limiting case for the protocol, now we determine the accuracy of this approximation. For this purpose we incorporate the fidelity of quantum states as a useful figure of merit. We compute it for the non-local $\chi$ matrix and its MA. Since, the MA gives a trivial state, one can easily compute the fidelity \cite{zyczkowski05, nielsen02}
\begin{equation}
    F\left(\chi, \chi_{MA}\right) = \frac{1}{4^{N}}\tr(\sqrt{\chi})^2=\frac{1}{4^N},
\end{equation}
where we used the fact that $\chi$ is a projector (i.e. $\sqrt{\chi}=\chi$, and $\tr(\chi) = 1$). As mentioned before, this result represents the worst case scenario, and is unlikely to happen in real experiments (especially if one is interested in low depth circuits), where hardware building blocks operate on fairly high gate fidelity (95-99\%, with lower fidelities for multi-qubit gates, and higher for single-qubit ones). Therefore, we can escape this unfavorable scaling by restricting to channels that are close to perfect (noiseless) case, i.e.~to the identity channel ($\chi_{00}=1$ and all other elements equal to zero). This also implies that non-local effects are comparably small to the leading order, which is predominantly determined by the $\chi_{00}$ element. Similar restrictions are commonly considered in benchmarking literature (see for example \cite{2020FlammiaEfficient}), since they represent noise regimes that are more relevant for the current hardware technology and help tailor error correcting schemes.  
In order to properly address this issue, we incorporate numerical methods to find out how well the MA can represent the true non-local noise process. Here, we use random sampling of full $\chi$ matrices and random samples of Pauli channels (i.e. $\chi$ matrices with a random probability vector on the diagonal and all other elements equal to zero). For the case of the full random $\chi$ processes, we explore systems composed of $N = 2, 3, 4$ qubits, while for Pauli channels we additionally look at $N = 5$. The results are displayed in~Fig.~\ref{fig:non_local_vs_MA}, where we took 10,000 samples of random channels (generated with QuTiP \cite{johansson12}), and computed all marginals of the multi-qubit $\chi$ matrix (i.e. tracing out all but one qubit) and compared fidelity between a tensor product of the marginals (essentially what we call the MA) and the non-local one.
\begin{figure*}[htb]
    \centering
    \includegraphics[width=0.48\textwidth]{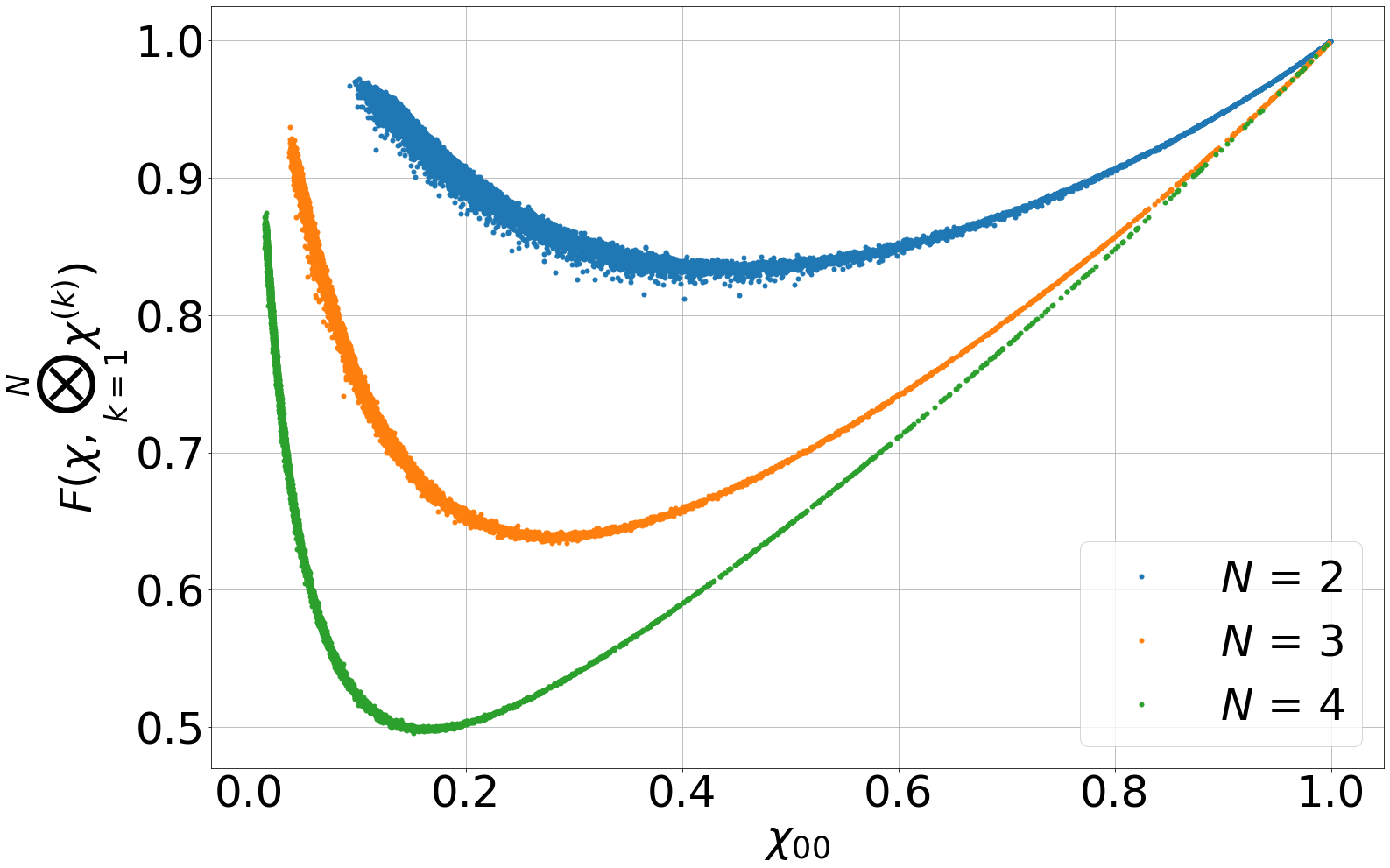}
    \includegraphics[width=0.48\textwidth]{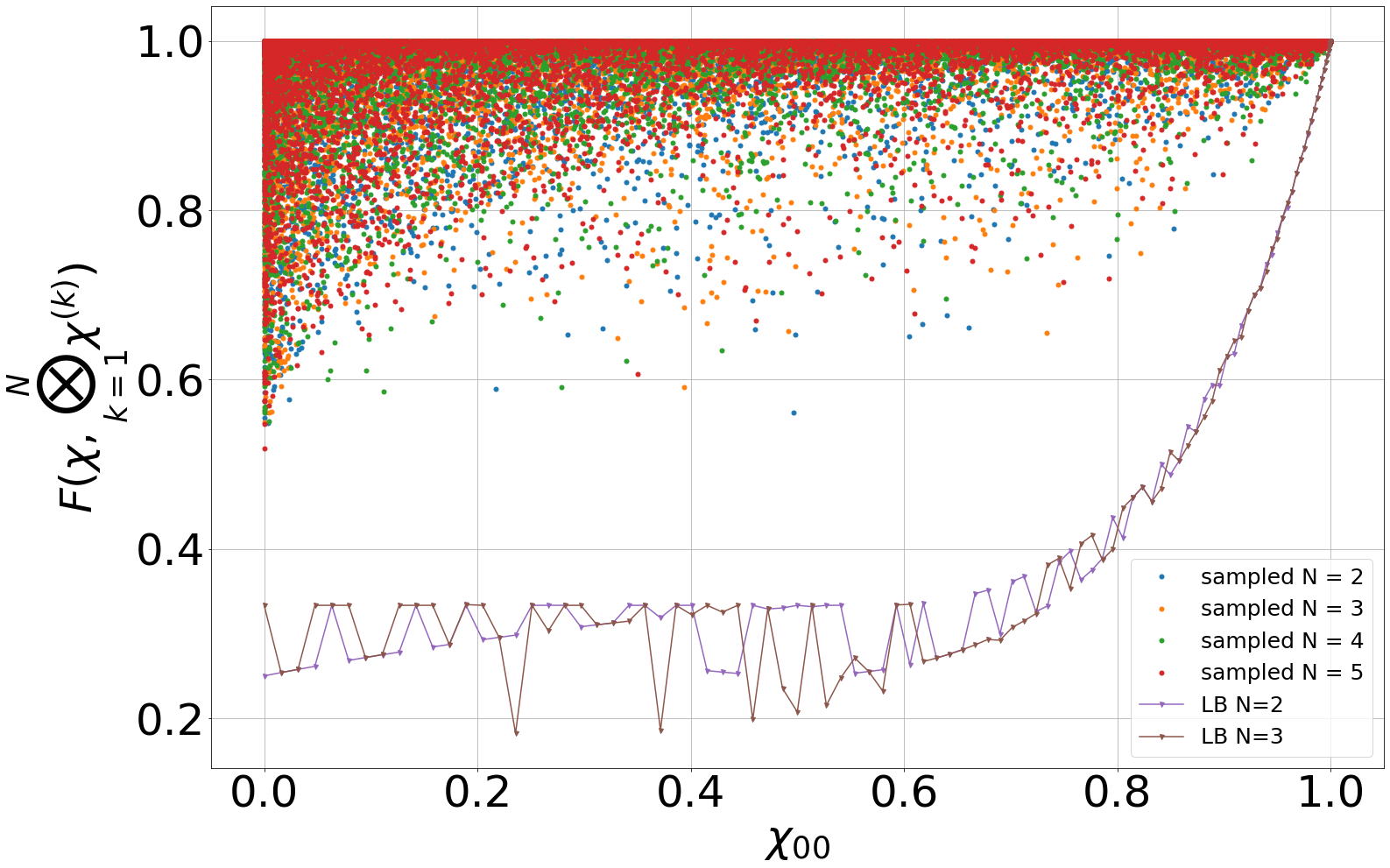}
    \caption{Comparison between random non-local $\chi$ matrix and its MA in terms of fidelity. Left plot depicts case of full random channels, while right plot restricts the analysis to random Pauli channels. Additionally, we provide a numerical lower bound (LB) on Pauli channels with two and three qubits (right plot).}
    \label{fig:non_local_vs_MA}
\end{figure*}

For random Pauli channels, we additionally numerically minimize the fidelity between the $4^N$ dimensional probability vector representing the Pauli channel, and its MA \footnote{The marginal approximation for Pauli channels is also done on the level of matrices, and not on the probability vectors.}. In order to guarantee a genuine probability distribution over our parameters (without having to impose any constraint) we use the modified Hurwitz parametrization for the probability vector \cite{Hurwitz1897,Zyczkowski2001}
\begin{align}
    \chi = &\mathrm{diag}[ \cos^2(\theta_{4^N-1}), \cos^2(\theta_{4^N-2})\sin^2(\theta_{4^N-1}),\ldots,&\nonumber  \\
    &\sin^2(\theta_1)\sin^2(\theta_2),\cdots\sin^2(\theta_{4^N-1}) ] &. 
\end{align}
We employ Sequential Least Squares Programming (SLSQP) \cite{kraft88} optimization routine to find the lower bound. Surprisingly, two and three qubit channels display similar lower bounds (in particular for high fidelity channels, i.e.~$\chi_{00}$ close to one).
The key observation is that for channels with reasonably large $\chi_{00}$ (corresponding to the identity channel), which is directly related to the gate/circuit fidelity, the MA can provide results with acceptable accuracy. Therefore, the MATEN protocol, identifies a MA that can estimate the leading order of the effective noise channel.

\section{Results}\label{results}

In this section we present the success of the method presented in \ref{singlequbitnoisecharacterization} for noisy simulations.

\subsection{Classical Simulation}\label{classicalsimulation}

For classical simulations, we test our characterization method against a variety of noise sources. Noiseless and noisy classical simulations are performed via pure state and density matrix simulations with HybridQ, an open-source hybrid quantum simulator \cite{mandra21}. In some cases, we additionally generate and apply error channels via  QuTiP \cite{johansson12} an open-source toolbox that allows for classical simulation of open quantum systems. With this capability of finding ideal and noisy states and operators, we can easily compute metrics needed to evaluate our method. For all of these experiments, we test the characterization method on parameterized QAOA circuits for QUBO problems.

\subsubsection{Purely Local Noise}\label{purelylocalnoise}
First we test the efficacy of the characterization method laid out in Sec.~\ref{singlequbitnoisecharacterization} for predicting $\chi$ matrices that we manually apply at the end of noiseless classical simulation. To do this, we pick a $\chi_{in}$ matrix by iteratively selecting elements uniformly randomly from the interval $[0,1]$ for the elements $\vec{p}$, and $[-1,1]$ for $\vec{t}$ and $\vec{v}$ in Eq.~\eqref{eq:chi_matrix}, and checking if the resultant map is physical (i.e. $\chi\ge0$) until we succeed.  We further choose the same $\chi_{in}$ matrix on each qubit, although this is relaxed in the next section. We additionally choose random QUBO problems by randomly drawing $J$ and $h$ from a uniform distribution in range $[0,1]$. In this experiment, we should expect that for some reasonable number of parameter settings (size of $\mathbb{S}$) and for a sufficient number of shots (measurements), we should be able exactly recover the input $\chi_{in}$ matrix to arbitrary precision, as the noise is taken to fit perfectly within the MATEN approximation. We quantify the accuracy of determining $\chi_{in}$ by taking the L2 distance between the elements of $\chi_{in}$ and $\chi_{pred}$, the process matrix our method predicts. The results for various values of shot number and number of regression angles are plotted in~Fig.~\ref{fig:purely_local_sim}. These plots are generated using statevector simulations for perfect evaluation of observables.

\begin{figure}[htb]
    \centering
    \includegraphics[width=0.48\textwidth]{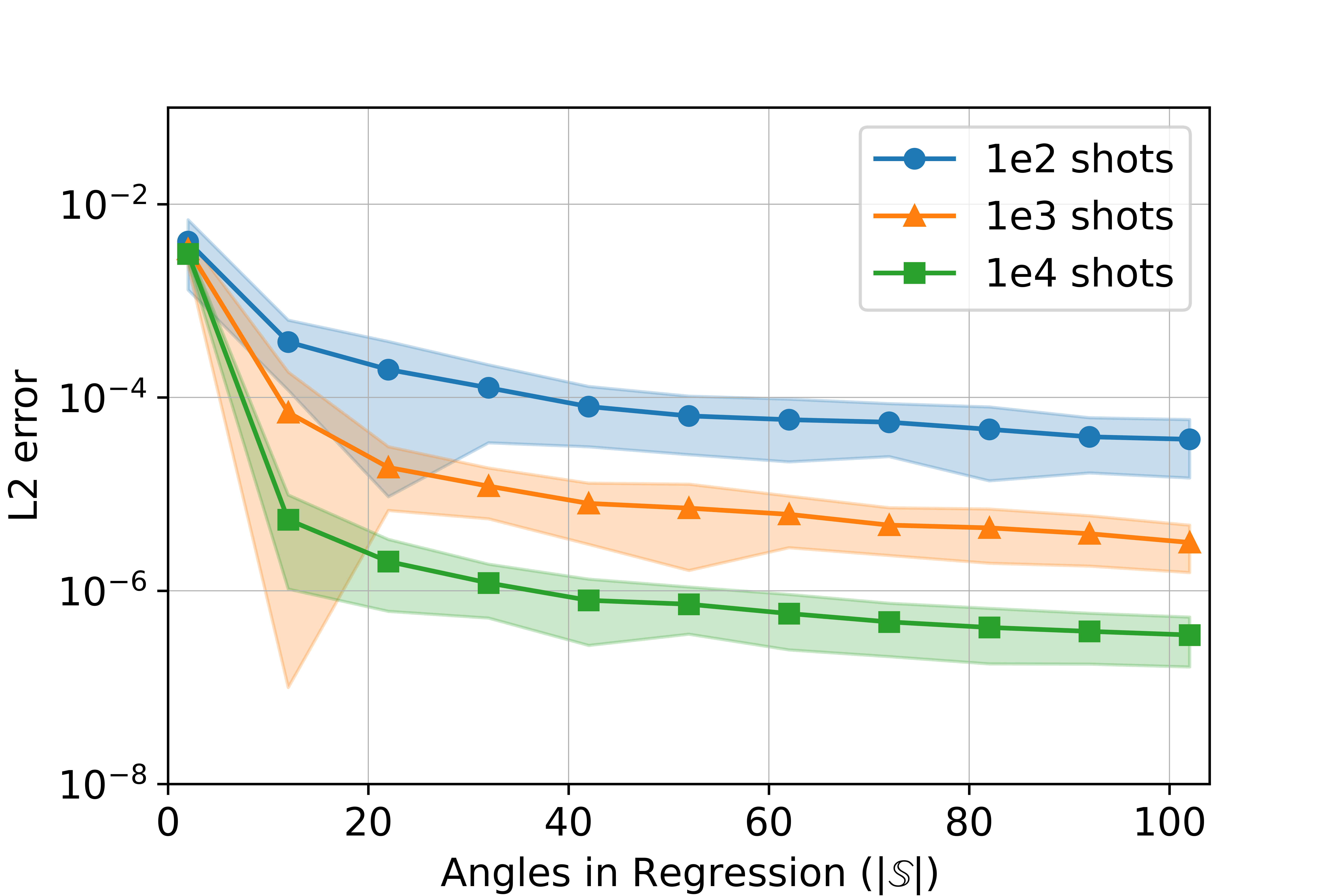}
    \caption{Average L2 distance between the true and predicted $\chi$ for classical simulations for randomly chosen two qubit QUBO instances with weights in the range $[0,1]$, as a function of the number of angles $|\mathbb{S}|$ used in the regression and the number of shots used in the estimation of expectation values. Solid lines depict the average over 100 runs, shading depicts one standard deviation above and below the average. For large number of shots and angles, the distance is below $10^{-6}$.}
    \label{fig:purely_local_sim}
\end{figure}

Indeed, we see that for a typical case, increasing the number of angles and the number of shots used in regression allows for more accurate determination of $\chi_{in}$. We further see from this figure that the L2 distance shrinks with added number of shots, by roughly a factor of $10$ when the number of shots increases by a factor of $10$. We later numerically see this roughly polynomial scaling with the number of shots for various values of $|\mathbb{S}|$. For instance, with $|\mathbb{S}|=16$ we find the L2 distance goes as $numshots^{-1.39}$. We note that the L2 distance between randomly chosen $\chi$ matrices was numerically found to be $.800 \pm .125$, but we see fidelities much higher than this value for sufficiently large $|\mathbb{S}|$ and number of shots.

\subsubsection{Non-Local Noise At End of Circuit}\label{nonlocalnoise}

In order to test the resiliency of the noise characterization procedure, we must test the method against noise models that a MATEN is not suited to perfectly capture. For the first of these models, we choose a constant error channel that exists only at the end of a quantum circuit, but is not a simple tensor product of single qubit channels. In order to apply this combination of local and non-local noise, then, we apply an error map of the form
\begin{equation} \label{eq:correlated_channel}
     \mathcal{E} = (1-c) \mathcal{E}_1^{(n)} + c \mathcal{E}_n^{(n)}, 
\end{equation}
where we have a combination of purely local channel  $\mathcal{E}_1^{(n)}$ and nonlocal channel $\mathcal{E}_n^{(n)}$ weighted by a correlation factor $c\in[0,1]$. In this section we allow the $\chi_{in}$ matrices to vary for each qubit (in $\mathcal{E}_1^{(n)}$). Since due to the addition of extra noise ($\mathcal{E}_n^{(n)}$) we no longer expect that $\chi_{pred} \approx \chi_{in}$, we no longer report the fidelity between the two. Instead, we use $(1/3)*r(\braket{X},\braket{\tilde{X}})+r(\braket{Y},\braket{\tilde{Y}})+r(\braket{Z},\braket{\tilde{Z}})$, the average Pearson correlation coefficient $r(x,y)$ between the measured and predicted expectation values of $\braket{X}$, $\braket{Y}$, and $\braket{Z}$, as these tell us how well our noise model predicts simple observables of interest on the quantum device. However, it is possible to induce overfitting, especially when the number of considered parameter settings ($|\mathbb{S}|$) is small. Thus we additionally look at correlations for an additional ``testing set'' of parameter settings. For our experiments at around $|\mathbb{S}|=50$ however, these correlations very closely matched that of the training set, so we only present correlations of the testing set for the following cases. In addition to correlation, we also use Choi fidelity, defined as $F(\Phi_1, \Phi_2)$, the state fidelity between Choi matrices $\Phi_1, \; \Phi_2$, representing respectively the entire $n$-qubit maps generated from the chosen error channels, and the predicted MATEN from our method. We present the results from the characterization of this noise model in Fig.~\ref{fig:nonlocal_sim}. For these experiments, we fix our problem Hamiltonian to a fully-connected QUBO instance with all $J=1$ and all $h=0$. Evolution and expectation values are evaluated using density matrix simulation.

\begin{figure*}[htb]
    \centering
    \includegraphics[width=0.48\textwidth]{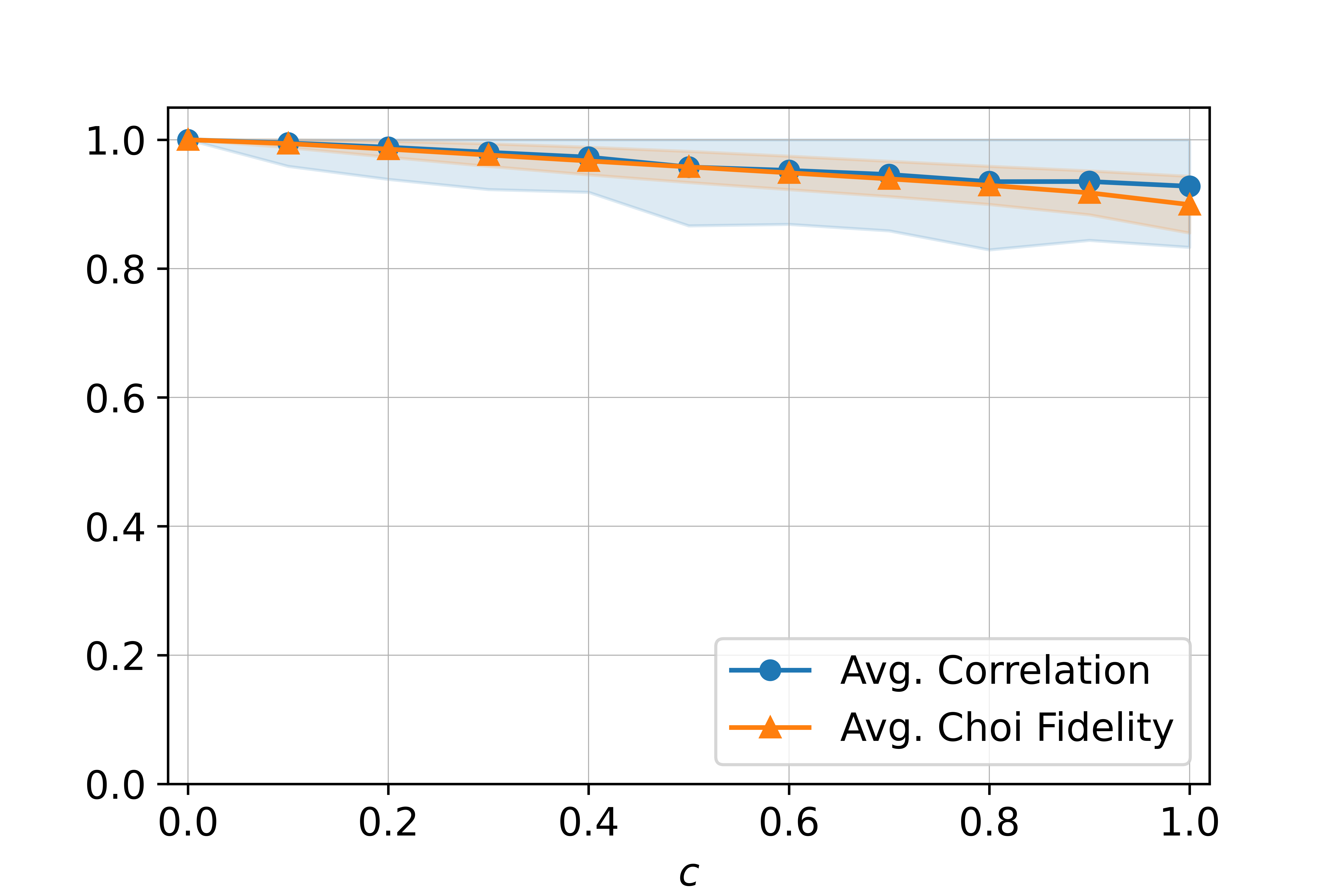}
    \includegraphics[width=0.48\textwidth]{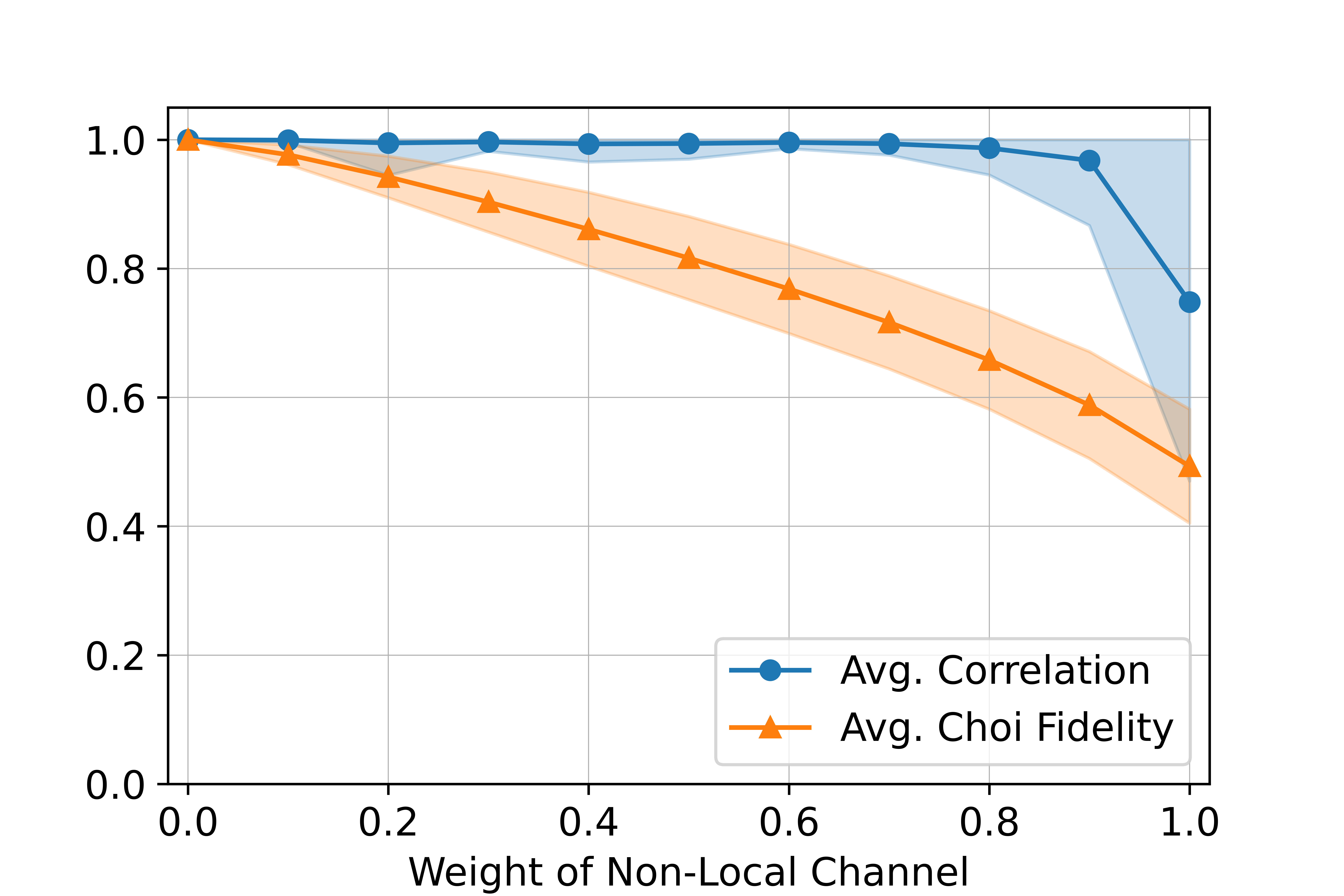}
    \caption{Average testing correlations and full state fidelities between the actual noise model and predicted noise model for two (left) and four (right) qubit fully connected QUBO problems with all $J=1$ and all $h=0$, as a function of the weight of the applied non-local channel, or $c$ in Eq.~\eqref{eq:correlated_channel}. Solid lines depict the average over 100 runs, shading depicts one standard deviation above and below the average.}
    \label{fig:nonlocal_sim}
\end{figure*}

From these simulations we see that, as expected, when $c=0$ and there is only local noise, the model works extremely well. However, as more non-local noise is added into the system, the ability to accurately predict the expectation values of Pauli observables begins to falter. At $c=1$, we typically see a sharp downturn of correlations, as at this point there we are not injecting any purely local noise to our system, thus weakening the accuracy of the MATEN.

\subsubsection{Sampling Noise}\label{samplingnoise}

In addition to the above tests, we also experimented with adding in sampling noise to our noisy simulations. To accomplish this, we choose random Gaussian perturbations with mean $0$ and standard deviation of $1/\sqrt{numshots}$ to add to all expectation value measurements, simulating the effect of sampling error on the evaluation of expectation values. Given this form of noise, we repeated analysis from above, running QAOA with cost Hamiltonian given by Eq.~\eqref{eq:localQAOAop} with two qubits and all $h=0$, $J=1$. We varied the number of shots on the x-axis, and the results of this setup are shown in Fig.~\ref{fig:twoq_nonlocal_sim_sweep}.

\begin{figure}[htb]
    \centering
    \includegraphics[width=0.48\textwidth]{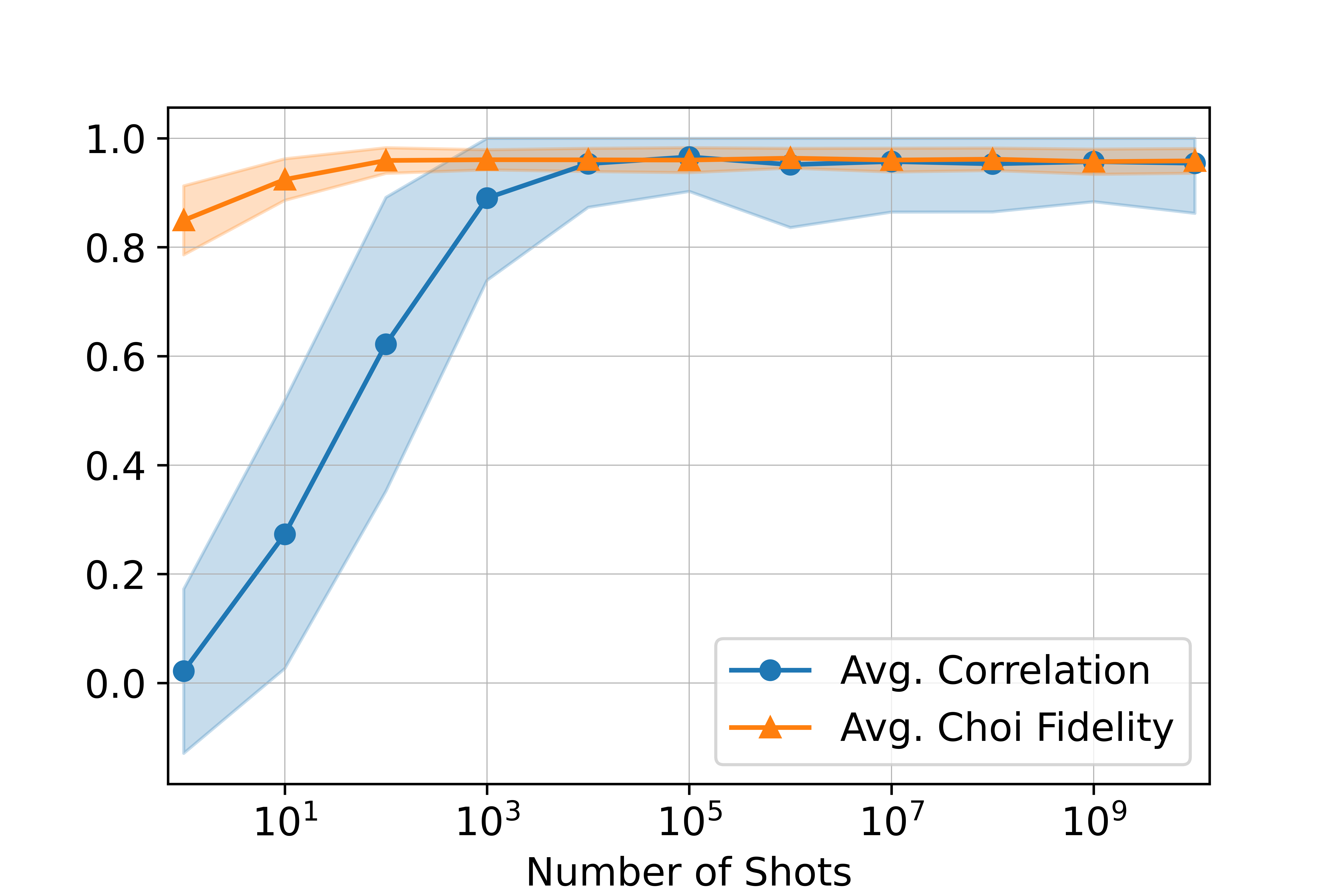}
        \caption{Average testing correlations and Choi fidelities between the actual noise model and predicted noise model for two qubit QUBO problem with $J=1$ and both $h=0$, as a function of the number of shots used to estimate expectation values. Solid lines depict the average over 100 runs, shading depicts one standard deviation above and below the average.}
    \label{fig:twoq_nonlocal_sim_sweep}
\end{figure}

From the simulations we can see that sampling noise diminishes the ability of the method to accurately fit the noisy measurements to ideal measurements, as well as predict the value of noisy measurements. The stochastic noise in causes expectation values to fluctuate between measurements, thus essentially introducing a non-constant noise model. This may cause poor performance as our method depends on having the same error channels for all angles and measurement bases. 

Additionally, we note that poor performance may arise if errors are angle-dependent, leading to an error model that is non-constant between different angles in a similar manner to sampling noise. Errors can be extremely angle-dependent on quantum computers, especially for parameterized two qubit gates such as \cite{abrams20}, so this feature could be an important limitation in the success of the method in the near term.

\subsubsection{Larger Systems: Phasing and Mixing Rotation Error}\label{rotationerror}

Finally, in order to scale our simulations to larger system sizes, we performed our characterization routine on $10$-qubit instances. For these runs, selecting and applying randomly generated $10$-qubit Kraus maps becomes numerically prohibitive, so we switch to a simpler and more realistic noise model. For these experiments, we assume that there is some stochastic error, or deviation in the parameters for both the phasing and mixing operators. In particular, the QAOA angles are assumed to be normally distributed about the desired mean value, with a non-zero standard deviation that defines the total amount of noise. For a given phasing gate $e^{-i\gamma Z_iZ_j}$, this noise is introduced through the Kraus operators
\begin{equation}
\label{eq:ZZ_noise}
    A_1 = \sqrt{1-\omega}I,\  A_2 = \sqrt{\omega}Z_iZ_j,
\end{equation}
and for a mixing gate $e^{-i\beta X_i}$ we apply 
\begin{equation}
\label{eq:X_noise}
    A_1 = \sqrt{1-\omega}I,\ A_2 = \sqrt{\omega}X_i.
\end{equation}
Here $\omega$ defines the amount of noise (related to the standard deviation in the angles' values).
A derivation of these noise models is shown in Appendix Sec \ref{non_constant_mixing_overrotations}.
This model applies the two-qubit dephasing noise layer (Eq.~\eqref{eq:ZZ_noise}) on each pair qubits the phase gates act, after the dephasing unitaries and directly before the mixing layer. After the mixing layer the one qubit $X$ noise is applied on each qubit (Eq.~\eqref{eq:X_noise}). 

Under this noise model we can test our characterization method on larger systems, and test against the assumption that all noise is applied at the end of the circuit. We display the results for the method on $10$-qubit ring and fully-connected QUBO problems in Fig.~\ref{fig:overrots}.

\begin{figure*}[htb]
    \centering
    \includegraphics[width=0.48\textwidth]{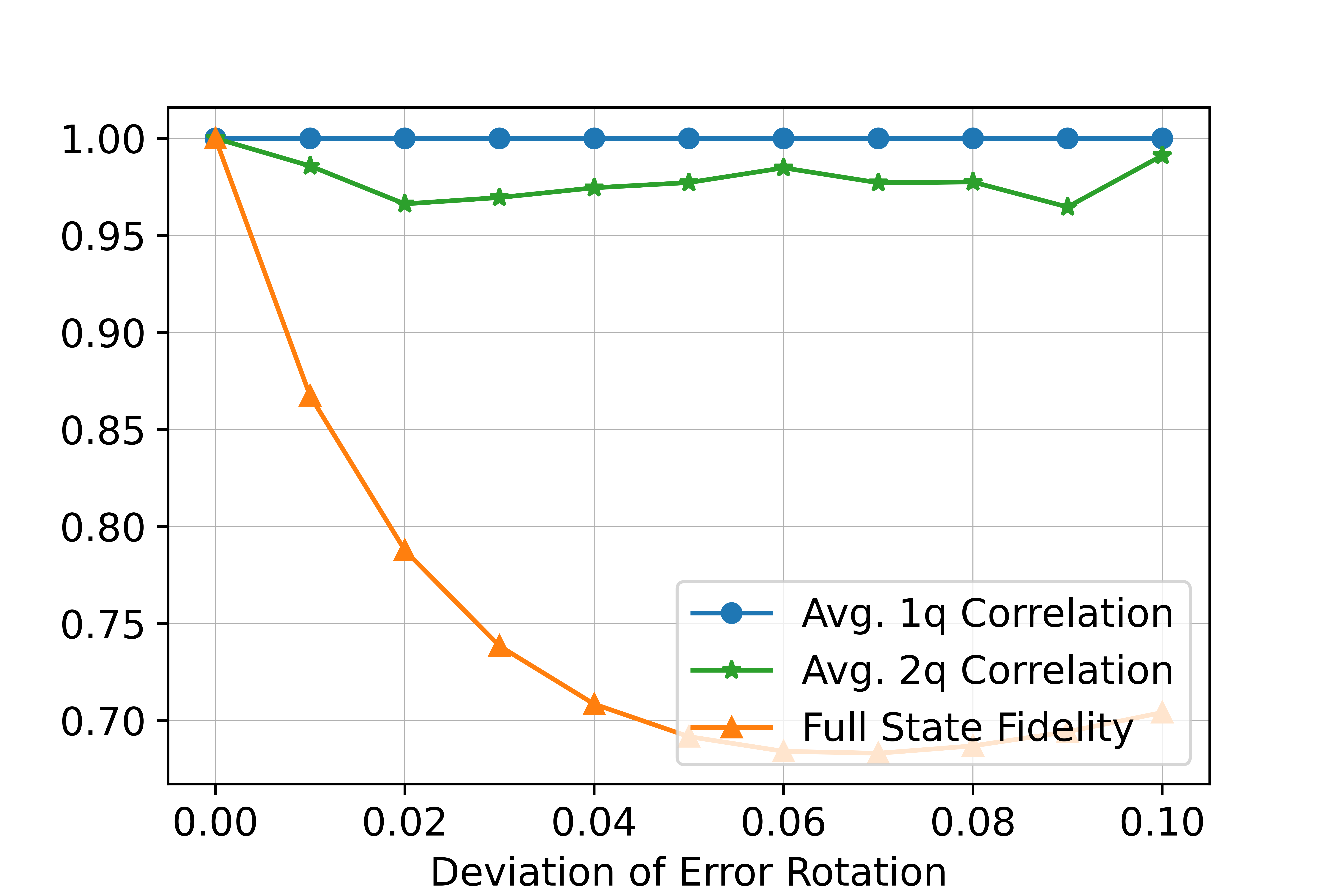}
    \includegraphics[width=0.48\textwidth]{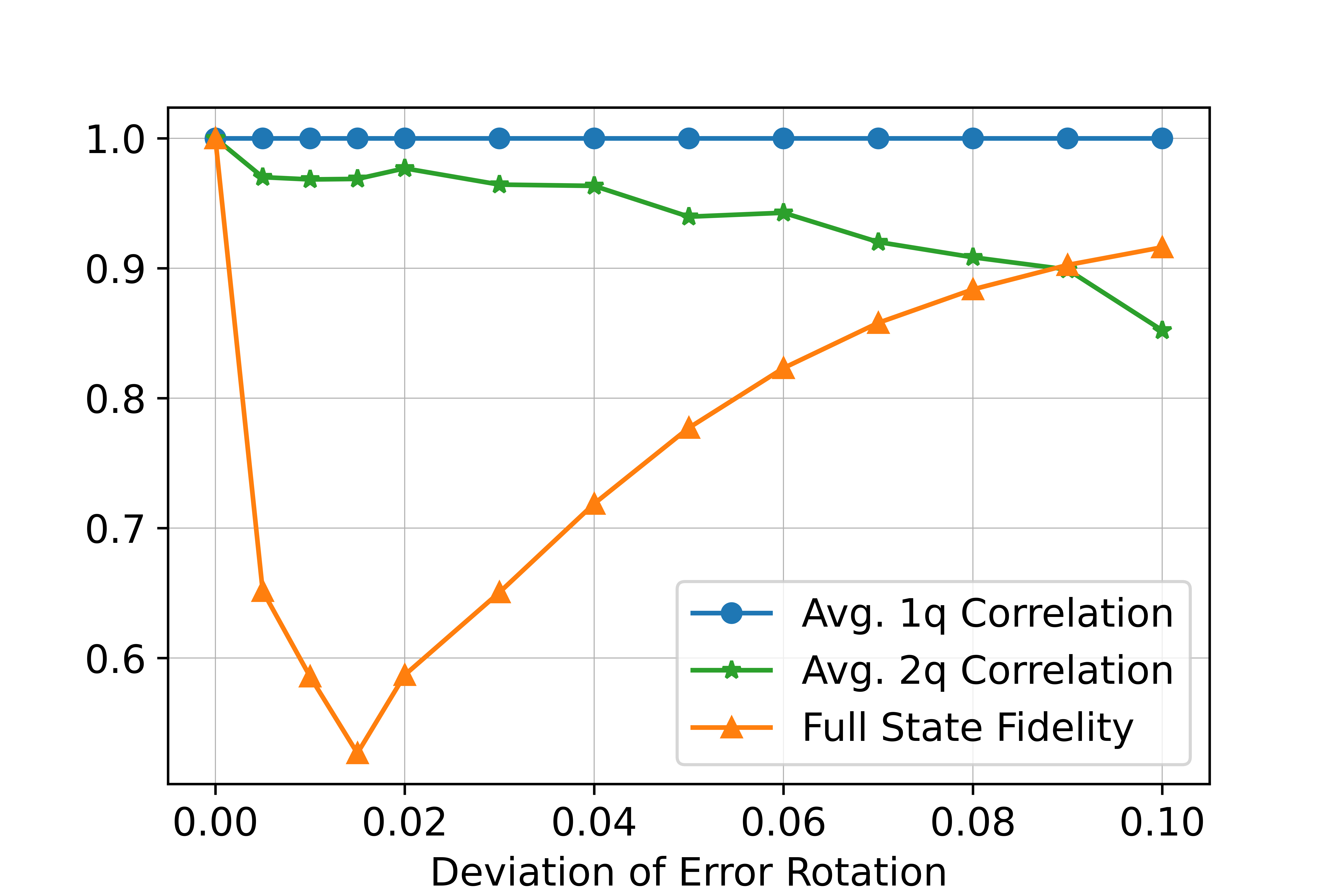}
    \caption{Average testing correlations and Choi fidelities between the actual noise model and predicted noise model for 10-qubit ring (right) and fully-connected (right) QUBO problems with all $J=1$ and all $h=0$, as a function of the deviation $\omega$ of both the phasing and mixing operators.}
    \label{fig:overrots}
\end{figure*}
For these plots, no matter the value of $\omega$, we saw that we were able to perfectly reproduce $1$-qubit correlations, so we chose to add in the average of all $2$-qubit correlations as well. Additionally, we report the average fidelity between the actual $10$-qubit noisy density matrix and the predicted density matrix using the characterized noise model. From these results we find that the fidelity drops rapidly, especially for the fully connected case. Crucially, however, the $1$ and $2$-qubit correlations remain very high, even as the $\omega$ grows. We note that on the fully connected plot, the fidelity rises after $\omega \approx .02$. This is likely explained by the fact that $\omega=0.5$ corresponds to the maximally dephasing channel, which our model can capture well. Thus we expect to see the fidelity drop initially as $\omega$ grows, then rise back to $1$ when $\omega=.5$, and then follow a symmetric pattern once $\omega>.5$. From these results, however, our main takeaway is that even in the presence of noise which is not local and not strictly at the end of the circuit, the method finds a suitable MATEN approximation that is able to replicate single-qubit expectation values perfectly and two-qubit expectation values very well, even as we scale to large system sizes.

\section{Characterization of Rigetti's Aspen-9 Device}\label{characterizationdevice}
In this section we apply the error characterization method from Sec.~\ref{singlequbitnoisecharacterization} to the Aspen-9 Quantum Processing Unit (QPU) from Rigetti Computing \cite{rigetti}. We run the characterization procedure for QAOA circuits with phase separation given by Hamiltonians of the form in Eq.~\eqref{eq:localQAOAop}, with all $h_i=1$ (to break $\mathbb{Z}_2$ symmetry) and $J_{ij}=\delta_{i+1, j}$ (forming a line topology), and implemented using a single CPHASE($\gamma$) gate, and with mixing via the standard X-mixer. These experiments were run at $N=2$ and $N=6$ with $|\mathbb{S}|=100$, where $N$ is the number of qubits and $|\mathbb{S}|$ is the number of different parameter settings used. For these experiments, we run under three cases.
\begin{enumerate}[label=\alph*)]
    \item (1q only) Remove all two-qubit (CPHASE) gates (equivalent to setting all $J_{ij}$ to $0$). The intention of this is to make sure that our method works when only single qubit gates are present, removing main sources of crosstalk and non-local noise, which could distort the results. 
    \item (2q idle) Add back in two qubit (CPHASE) gates, but set the angles ($\gamma$) of all two-qubit gates to $0$ (again equivalent to setting all $J_{ij}$ to $0$). This ideally implements the same circuit as the previous case, but two-qubit gates are physically implemented in the circuit.
    \item (2q active) Lift the restriction of setting two-qubit gate angles to $0$, thus performing the method completely as intended.
\end{enumerate}

For these experiments, much like Sec.~\ref{nonlocalnoise}, we present statistics on the correlations between predicted and observed Pauli expectation values. These are shown for both the two and six qubit cases in table \ref{tab:qpu_2q}.
\begin{table}[htb!]
\centering
\begin{tabular}{ | c || c | c | c|} 
\hline
  qubit & 1q only & 2q idle & 2q active  \\
 \hline
 34 & 0.9964 & 0.5154 & 0.9622 \\ 
 35 & 0.9982 & 0.7758 & 0.9704 \\

 \hline 
\end{tabular}
\caption{Correlations for the method performed on QAOA circuits for qubits 34 and 35 on Rigetti Aspen-9 device. The low values for the 2q idle case are likely explained by day-to-day changes in qubit calibration and error sources on the device.}
\label{tab:qpu_2q}
\end{table}

For the two-qubit experiments, we see that the method is able to predict expectation values of all Pauli observables with a high correlation to the experimental values. We note that there is low fidelity for the ``2q idle" case. This is likely explained by the fact that this case was run a few days after the other experiments, as this experiment idea was conceived after running the ``1q only" and ``2q active" cases. Due to day-to-day changes in calibration, qubits 34 and 35 may have experienced calibration issues on the day of running. Unfortunately, the Aspen-9 device was de-commissioned before a re-run of the experiment began. This faulty qubit can result in elevated angle-dependent, nonlocal, or inter-circuit noise, which have the potential to reduce the accuracy of a MATEN.

The six-qubit experiments are presented in Table \ref{tab:qpu_6q}.

\begin{table}[htb!]
\centering
\begin{tabular}{ | c || c | c | c|} 
\hline
  qubit & 1q only & 2q idle & 2q active  \\
 \hline
 30 & 0.9963 & 0.9254 & 0.5477 \\ 
 31 & 0.9901 & 0.9041 & 0.6692 \\
 32 & 0.9923 & 0.6584 & 0.5363 \\ 
 33 & 0.9948 & 0.0480 & 0.0449 \\ 
 34 & 0.9936 & 0.8674 & 0.7855 \\ 
 35 & 0.9985 & 0.9908 & 0.9801 \\ 

 \hline 
\end{tabular}
\caption{Correlations for the method performed on QAOA circuits for qubits 30-35 on Rigetti Aspen-9 device.}
\label{tab:qpu_6q}
\end{table}

Here, we see that all metrics remain high for the ``1q only" case, but for the ``2q idle" case for qubits 32 and 33 we see a significant drop in regression score and average correlation. In the ``2q active" case, we see a further decline in the correlations of qubits 30, 31, 32, and 34. For this case, which matches most closely the type of experiments we would like to characterize, our method gives an average expectation value correlation of $.69 \pm .30$. These values are far from the ideal values of 1, but the positive correlation values suggest that the method approximately captures the dominant error channels present on the QPU. The wide variability in performance on various qubits suggest that certain qubits may have more angle-dependent noise, or may have larger sources of crosstalk, as analyzed in Sec.~\ref{classicalsimulation}. In particular, qubits 32 and 33 experience a sharp decline in correlations in both the ``2q idle/active" cases, with the correlations of qubit 33 plummeting to roughly $.05$. The correlations on qubit 33 of roughly $.05$ are additionally much lower than we see even on the right side of Fig.~\ref{fig:nonlocal_sim} or anywhere in Fig.~\ref{fig:overrots}. This indicates that the errors introduced by two-qubit gates are in a sense worse than both of these cases. We suspect this may be due to the fact that the added two qubit gate, even with all angles set to zero, may introduce some significant crosstalk between the two qubits that is far from the intended phasing operation, which the MATEN is not equipped to accurately handle. In the simulations we perform, artificially added errors come in the form of randomly chosen Kraus maps or overrotations, but the error maps on a quantum device may be of a specific, more detrimental for. Additionally, even with a two-qubit gate with angle set to $0$, it can be the case that a different unitary is applied from shot-to-shot, approaching the case of Sec.~\ref{samplingnoise}, which is the only source of noise we found to reduce correlations to such a low number. Thus we suspect that this error or shot-dependent noise may play a role in the extremely low correlations, as we would not expect to be able to accurately characterize any noise procedure that is changing over time.

\section{Discussion}\label{discussion}
In this paper we introduce the dual map framework for computing the effects of error maps on expectation values evaluated on a quantum computer. We then presented a method to compute a marginal approximation to the effective noise (MATEN) of a parameterized quantum circuit, that is efficient in terms of number of measurements needed to perform on a quantum computer and is simple to implement. We demonstrate that the method effectively computes a MATEN for local noise at the end of a circuit, and demonstrate that it can be effective even in the presence of nonlocal and inter-circuit noise, especially when the noise is only weakly correlated. We finally show that the method is effective in computing a MATEN on a few qubits of Rigetti's Aspen-9 quantum computer. 
Lower values in extracted correlations of expectation values can be inform us  that the system exhibits a fair amount of angle dependant (gate) noise, as well as errors that are absent in the theoretical model, e.g. readout or leakage to the non-computational subspace. The latter can be modelled in a similar fashion as qubits under our scheme, with the difference that the $\chi$ process matrix now needs to represent a qudit process. This, however, introduces an extra layer of complexity, which we leave for future analysis.

The error characterization method can additionally be used as a proxy for the fidelity of a gate, layer, or entire circuit, as the values of the computed $\chi$ matrices for each qubit (specifically $\chi_{00}$) quantify the difference between the ideal and noisy evolution. Furthermore, the returned $\chi$ can inform dominant sources of error, which can in turn point to particularly effective strategies from error mitigation, leading to algorithmic improvements on NISQ devices. Once dominant sources of error are determined, we leave these error-specific mitigation approaches as open problems for the reader.

The dual map framework introduced can be used to understand which error channels can be specifically detrimental for a circuit. For instance, with QAOA, we show that depolarizing noise simply flattens the energy landscape, thus it does not affect the location of optimal parameters for the algorithm. However, error sources such as amplitude damping may introduce non-trivial behavior. The characterization procedure we introduce can be used to characterize error in NISQ devices, especially for shallow circuits in which the effective noise channels are expected to be non-correlated. Overall, the dual map picture for error channels provides a simple and elegant method for researching the interplay between quantum error and algorithms in the future, and our characterization approach can significantly aid hardware-aware algorithm design on today's devices.

\section{Acknowledgements}
This material is based upon work supported by the U.S. Department of Energy, Office of Science, National Quantum Information Science Research Centers, Superconducting Quantum Materials and Systems Center (SQMS) under contract number DE-AC02-07CH11359 through NASA-DOE interagency agreement SAA2-403602. All authors appreciate support from the NASA Ames Research Center. JS, JM, ZW, FW are thankful for support from NASA Academic Mission Services, Contract No. NNA16BD14C.

\bibliography{refs} 

\begin{thebibliography}{51}%
\makeatletter
\providecommand \@ifxundefined [1]{%
 \@ifx{#1\undefined}
}%
\providecommand \@ifnum [1]{%
 \ifnum #1\expandafter \@firstoftwo
 \else \expandafter \@secondoftwo
 \fi
}%
\providecommand \@ifx [1]{%
 \ifx #1\expandafter \@firstoftwo
 \else \expandafter \@secondoftwo
 \fi
}%
\providecommand \natexlab [1]{#1}%
\providecommand \enquote  [1]{``#1''}%
\providecommand \bibnamefont  [1]{#1}%
\providecommand \bibfnamefont [1]{#1}%
\providecommand \citenamefont [1]{#1}%
\providecommand \href@noop [0]{\@secondoftwo}%
\providecommand \href [0]{\begingroup \@sanitize@url \@href}%
\providecommand \@href[1]{\@@startlink{#1}\@@href}%
\providecommand \@@href[1]{\endgroup#1\@@endlink}%
\providecommand \@sanitize@url [0]{\catcode `\\12\catcode `\$12\catcode
  `\&12\catcode `\#12\catcode `\^12\catcode `\_12\catcode `\%12\relax}%
\providecommand \@@startlink[1]{}%
\providecommand \@@endlink[0]{}%
\providecommand \url  [0]{\begingroup\@sanitize@url \@url }%
\providecommand \@url [1]{\endgroup\@href {#1}{\urlprefix }}%
\providecommand \urlprefix  [0]{URL }%
\providecommand \Eprint [0]{\href }%
\providecommand \doibase [0]{http://dx.doi.org/}%
\providecommand \selectlanguage [0]{\@gobble}%
\providecommand \bibinfo  [0]{\@secondoftwo}%
\providecommand \bibfield  [0]{\@secondoftwo}%
\providecommand \translation [1]{[#1]}%
\providecommand \BibitemOpen [0]{}%
\providecommand \bibitemStop [0]{}%
\providecommand \bibitemNoStop [0]{.\EOS\space}%
\providecommand \EOS [0]{\spacefactor3000\relax}%
\providecommand \BibitemShut  [1]{\csname bibitem#1\endcsname}%
\let\auto@bib@innerbib\@empty
\bibitem [{\citenamefont {Piveteau}\ \emph {et~al.}(2021)\citenamefont
  {Piveteau}, \citenamefont {Sutter}, \citenamefont {Bravyi}, \citenamefont
  {Gambetta},\ and\ \citenamefont {Temme}}]{err-mitt-EC}%
  \BibitemOpen
  \bibfield  {author} {\bibinfo {author} {\bibfnamefont {Christophe}\
  \bibnamefont {Piveteau}}, \bibinfo {author} {\bibfnamefont {David}\
  \bibnamefont {Sutter}}, \bibinfo {author} {\bibfnamefont {Sergey}\
  \bibnamefont {Bravyi}}, \bibinfo {author} {\bibfnamefont {Jay~M.}\
  \bibnamefont {Gambetta}}, \ and\ \bibinfo {author} {\bibfnamefont {Kristan}\
  \bibnamefont {Temme}},\ }\bibfield  {title} {\enquote {\bibinfo {title}
  {Error mitigation for universal gates on encoded qubits},}\ }\href {\doibase
  10.1103/PhysRevLett.127.200505} {\bibfield  {journal} {\bibinfo  {journal}
  {Phys. Rev. Lett.}\ }\textbf {\bibinfo {volume} {127}},\ \bibinfo {pages}
  {200505} (\bibinfo {year} {2021})}\BibitemShut {NoStop}%
\bibitem [{\citenamefont {Temme}\ \emph {et~al.}(2017)\citenamefont {Temme},
  \citenamefont {Bravyi},\ and\ \citenamefont {Gambetta}}]{zne1}%
  \BibitemOpen
  \bibfield  {author} {\bibinfo {author} {\bibfnamefont {Kristan}\ \bibnamefont
  {Temme}}, \bibinfo {author} {\bibfnamefont {Sergey}\ \bibnamefont {Bravyi}},
  \ and\ \bibinfo {author} {\bibfnamefont {Jay~M.}\ \bibnamefont {Gambetta}},\
  }\bibfield  {title} {\enquote {\bibinfo {title} {Error mitigation for
  short-depth quantum circuits},}\ }\href {\doibase
  10.1103/PhysRevLett.119.180509} {\bibfield  {journal} {\bibinfo  {journal}
  {Phys. Rev. Lett.}\ }\textbf {\bibinfo {volume} {119}},\ \bibinfo {pages}
  {180509} (\bibinfo {year} {2017})}\BibitemShut {NoStop}%
\bibitem [{\citenamefont {Li}\ and\ \citenamefont {Benjamin}(2017)}]{zne2}%
  \BibitemOpen
  \bibfield  {author} {\bibinfo {author} {\bibfnamefont {Ying}\ \bibnamefont
  {Li}}\ and\ \bibinfo {author} {\bibfnamefont {Simon~C.}\ \bibnamefont
  {Benjamin}},\ }\bibfield  {title} {\enquote {\bibinfo {title} {{Efficient
  Variational Quantum Simulator Incorporating Active Error Minimization}},}\
  }\href {\doibase 10.1103/PhysRevX.7.021050} {\bibfield  {journal} {\bibinfo
  {journal} {Phys. Rev. X}\ }\textbf {\bibinfo {volume} {7}},\ \bibinfo {pages}
  {021050} (\bibinfo {year} {2017})}\BibitemShut {NoStop}%
\bibitem [{\citenamefont {Nielsen}(2002)}]{nielsen02}%
  \BibitemOpen
  \bibfield  {author} {\bibinfo {author} {\bibfnamefont {Michael~A}\
  \bibnamefont {Nielsen}},\ }\bibfield  {title} {\enquote {\bibinfo {title} {A
  simple formula for the average gate fidelity of a quantum dynamical
  operation},}\ }\href {\doibase 10.1016/s0375-9601(02)01272-0} {\bibfield
  {journal} {\bibinfo  {journal} {Physics Letters A}\ }\textbf {\bibinfo
  {volume} {303}},\ \bibinfo {pages} {249–252} (\bibinfo {year}
  {2002})}\BibitemShut {NoStop}%
\bibitem [{\citenamefont {Wudarski}\ \emph {et~al.}(2020)\citenamefont
  {Wudarski}, \citenamefont {Marshall}, \citenamefont {Petukhov},\ and\
  \citenamefont {Rieffel}}]{wudarski20}%
  \BibitemOpen
  \bibfield  {author} {\bibinfo {author} {\bibfnamefont {Filip}\ \bibnamefont
  {Wudarski}}, \bibinfo {author} {\bibfnamefont {Jeffrey}\ \bibnamefont
  {Marshall}}, \bibinfo {author} {\bibfnamefont {Andre}\ \bibnamefont
  {Petukhov}}, \ and\ \bibinfo {author} {\bibfnamefont {Eleanor}\ \bibnamefont
  {Rieffel}},\ }\bibfield  {title} {\enquote {\bibinfo {title} {Augmented
  fidelities for single-qubit gates},}\ }\href {\doibase
  10.1103/PhysRevA.102.052612} {\bibfield  {journal} {\bibinfo  {journal}
  {Phys. Rev. A}\ }\textbf {\bibinfo {volume} {102}},\ \bibinfo {pages}
  {052612} (\bibinfo {year} {2020})}\BibitemShut {NoStop}%
\bibitem [{\citenamefont {Emerson}\ \emph {et~al.}(2005)\citenamefont
  {Emerson}, \citenamefont {Alicki},\ and\ \citenamefont
  {{\textbackslash}.Zyczkowski}}]{emerson05}%
  \BibitemOpen
  \bibfield  {author} {\bibinfo {author} {\bibfnamefont {Joseph}\ \bibnamefont
  {Emerson}}, \bibinfo {author} {\bibfnamefont {Robert}\ \bibnamefont
  {Alicki}}, \ and\ \bibinfo {author} {\bibfnamefont {Karol}\ \bibnamefont
  {{\textbackslash}.Zyczkowski}},\ }\bibfield  {title} {\enquote {\bibinfo
  {title} {Scalable noise estimation with random unitary operators},}\ }\href
  {\doibase 10.1088/1464-4266/7/10/021} {\ \textbf {\bibinfo {volume} {7}},\
  \bibinfo {pages} {S347--S352} (\bibinfo {year} {2005})},\ \bibinfo {note}
  {publisher: IOP Publishing}\BibitemShut {NoStop}%
\bibitem [{\citenamefont {Magesan}\ \emph {et~al.}(2011)\citenamefont
  {Magesan}, \citenamefont {Gambetta},\ and\ \citenamefont
  {Emerson}}]{magesan11}%
  \BibitemOpen
  \bibfield  {author} {\bibinfo {author} {\bibfnamefont {Easwar}\ \bibnamefont
  {Magesan}}, \bibinfo {author} {\bibfnamefont {J.~M.}\ \bibnamefont
  {Gambetta}}, \ and\ \bibinfo {author} {\bibfnamefont {Joseph}\ \bibnamefont
  {Emerson}},\ }\bibfield  {title} {\enquote {\bibinfo {title} {Scalable and
  {Robust} {Randomized} {Benchmarking} of {Quantum} {Processes}},}\ }\href
  {\doibase 10.1103/PhysRevLett.106.180504} {\bibfield  {journal} {\bibinfo
  {journal} {Physical Review Letters}\ }\textbf {\bibinfo {volume} {106}},\
  \bibinfo {pages} {180504} (\bibinfo {year} {2011})},\ \bibinfo {note}
  {publisher: American Physical Society}\BibitemShut {NoStop}%
\bibitem [{\citenamefont {Claes}\ \emph {et~al.}(2021)\citenamefont {Claes},
  \citenamefont {Rieffel},\ and\ \citenamefont {Wang}}]{claes21}%
  \BibitemOpen
  \bibfield  {author} {\bibinfo {author} {\bibfnamefont {Jahan}\ \bibnamefont
  {Claes}}, \bibinfo {author} {\bibfnamefont {Eleanor}\ \bibnamefont
  {Rieffel}}, \ and\ \bibinfo {author} {\bibfnamefont {Zhihui}\ \bibnamefont
  {Wang}},\ }\bibfield  {title} {\enquote {\bibinfo {title} {Character
  {Randomized} {Benchmarking} for {Non}-{Multiplicity}-{Free} {Groups} {With}
  {Applications} to {Subspace}, {Leakage}, and {Matchgate} {Randomized}
  {Benchmarking}},}\ }\href {\doibase 10.1103/PRXQuantum.2.010351} {\bibfield
  {journal} {\bibinfo  {journal} {PRX Quantum}\ }\textbf {\bibinfo {volume}
  {2}},\ \bibinfo {pages} {010351} (\bibinfo {year} {2021})},\ \bibinfo {note}
  {publisher: American Physical Society}\BibitemShut {NoStop}%
\bibitem [{\citenamefont {Erhard}\ \emph {et~al.}(2019)\citenamefont {Erhard},
  \citenamefont {Wallman}, \citenamefont {Postler}, \citenamefont {Meth},
  \citenamefont {Stricker}, \citenamefont {Martinez}, \citenamefont
  {Schindler}, \citenamefont {Monz}, \citenamefont {Emerson},\ and\
  \citenamefont {Blatt}}]{erhard19}%
  \BibitemOpen
  \bibfield  {author} {\bibinfo {author} {\bibfnamefont {Alexander}\
  \bibnamefont {Erhard}}, \bibinfo {author} {\bibfnamefont {Joel~J.}\
  \bibnamefont {Wallman}}, \bibinfo {author} {\bibfnamefont {Lukas}\
  \bibnamefont {Postler}}, \bibinfo {author} {\bibfnamefont {Michael}\
  \bibnamefont {Meth}}, \bibinfo {author} {\bibfnamefont {Roman}\ \bibnamefont
  {Stricker}}, \bibinfo {author} {\bibfnamefont {Esteban~A.}\ \bibnamefont
  {Martinez}}, \bibinfo {author} {\bibfnamefont {Philipp}\ \bibnamefont
  {Schindler}}, \bibinfo {author} {\bibfnamefont {Thomas}\ \bibnamefont
  {Monz}}, \bibinfo {author} {\bibfnamefont {Joseph}\ \bibnamefont {Emerson}},
  \ and\ \bibinfo {author} {\bibfnamefont {Rainer}\ \bibnamefont {Blatt}},\
  }\bibfield  {title} {\enquote {\bibinfo {title} {Characterizing large-scale
  quantum computers via cycle benchmarking},}\ }\href {\doibase
  10.1038/s41467-019-13068-7} {\bibfield  {journal} {\bibinfo  {journal}
  {Nature Communications}\ }\textbf {\bibinfo {volume} {10}},\ \bibinfo {pages}
  {5347} (\bibinfo {year} {2019})}\BibitemShut {NoStop}%
\bibitem [{\citenamefont {Flammia}\ and\ \citenamefont
  {Liu}(2011)}]{flammia11}%
  \BibitemOpen
  \bibfield  {author} {\bibinfo {author} {\bibfnamefont {Steven~T.}\
  \bibnamefont {Flammia}}\ and\ \bibinfo {author} {\bibfnamefont {Yi-Kai}\
  \bibnamefont {Liu}},\ }\bibfield  {title} {\enquote {\bibinfo {title} {Direct
  {Fidelity} {Estimation} from {Few} {Pauli} {Measurements}},}\ }\href
  {\doibase 10.1103/PhysRevLett.106.230501} {\bibfield  {journal} {\bibinfo
  {journal} {Physical Review Letters}\ }\textbf {\bibinfo {volume} {106}},\
  \bibinfo {pages} {230501} (\bibinfo {year} {2011})},\ \bibinfo {note}
  {publisher: American Physical Society}\BibitemShut {NoStop}%
\bibitem [{\citenamefont {Chuang}\ and\ \citenamefont
  {Nielsen}(1997)}]{chuang97}%
  \BibitemOpen
  \bibfield  {author} {\bibinfo {author} {\bibfnamefont {Isaac~L.}\
  \bibnamefont {Chuang}}\ and\ \bibinfo {author} {\bibfnamefont {M.~A.}\
  \bibnamefont {Nielsen}},\ }\bibfield  {title} {\enquote {\bibinfo {title}
  {Prescription for experimental determination of the dynamics of a quantum
  black box},}\ }\href {\doibase 10.1080/09500349708231894} {\bibfield
  {journal} {\bibinfo  {journal} {Journal of Modern Optics}\ }\textbf {\bibinfo
  {volume} {44}},\ \bibinfo {pages} {2455--2467} (\bibinfo {year}
  {1997})}\BibitemShut {NoStop}%
\bibitem [{\citenamefont {Greenbaum}(2015)}]{greenbaum15}%
  \BibitemOpen
  \bibfield  {author} {\bibinfo {author} {\bibfnamefont {Daniel}\ \bibnamefont
  {Greenbaum}},\ }\bibfield  {title} {\enquote {\bibinfo {title} {Introduction
  to {Quantum} {Gate} {Set} {Tomography}},}\ }\href
  {http://arxiv.org/abs/1509.02921} {\bibfield  {journal} {\bibinfo  {journal}
  {arXiv:1509.02921 [quant-ph]}\ } (\bibinfo {year} {2015})},\ \bibinfo {note}
  {arXiv: 1509.02921}\BibitemShut {NoStop}%
\bibitem [{\citenamefont {Schirmer}\ \emph {et~al.}(2004)\citenamefont
  {Schirmer}, \citenamefont {Kolli},\ and\ \citenamefont {Oi}}]{schirmer04}%
  \BibitemOpen
  \bibfield  {author} {\bibinfo {author} {\bibfnamefont {S.~G.}\ \bibnamefont
  {Schirmer}}, \bibinfo {author} {\bibfnamefont {A.}~\bibnamefont {Kolli}}, \
  and\ \bibinfo {author} {\bibfnamefont {D.~K.~L.}\ \bibnamefont {Oi}},\
  }\bibfield  {title} {\enquote {\bibinfo {title} {Experimental {Hamiltonian}
  identification for controlled two-level systems},}\ }\href {\doibase
  10.1103/PhysRevA.69.050306} {\bibfield  {journal} {\bibinfo  {journal}
  {Physical Review A}\ }\textbf {\bibinfo {volume} {69}},\ \bibinfo {pages}
  {050306} (\bibinfo {year} {2004})},\ \bibinfo {note} {publisher: American
  Physical Society}\BibitemShut {NoStop}%
\bibitem [{\citenamefont {Kimmel}\ \emph {et~al.}(2015)\citenamefont {Kimmel},
  \citenamefont {Low},\ and\ \citenamefont {Yoder}}]{kimmel15}%
  \BibitemOpen
  \bibfield  {author} {\bibinfo {author} {\bibfnamefont {Shelby}\ \bibnamefont
  {Kimmel}}, \bibinfo {author} {\bibfnamefont {Guang~Hao}\ \bibnamefont {Low}},
  \ and\ \bibinfo {author} {\bibfnamefont {Theodore~J.}\ \bibnamefont
  {Yoder}},\ }\bibfield  {title} {\enquote {\bibinfo {title} {Robust
  calibration of a universal single-qubit gate set via robust phase
  estimation},}\ }\href {\doibase 10.1103/PhysRevA.92.062315} {\bibfield
  {journal} {\bibinfo  {journal} {Physical Review A}\ }\textbf {\bibinfo
  {volume} {92}},\ \bibinfo {pages} {062315} (\bibinfo {year} {2015})},\
  \bibinfo {note} {publisher: American Physical Society}\BibitemShut {NoStop}%
\bibitem [{\citenamefont {Sun}\ and\ \citenamefont {Geller}(2020)}]{sun20}%
  \BibitemOpen
  \bibfield  {author} {\bibinfo {author} {\bibfnamefont {Mingyu}\ \bibnamefont
  {Sun}}\ and\ \bibinfo {author} {\bibfnamefont {Michael~R.}\ \bibnamefont
  {Geller}},\ }\bibfield  {title} {\enquote {\bibinfo {title} {Efficient
  characterization of correlated {SPAM} errors},}\ }\href
  {http://arxiv.org/abs/1810.10523} {\bibfield  {journal} {\bibinfo  {journal}
  {arXiv:1810.10523 [quant-ph]}\ } (\bibinfo {year} {2020})},\ \bibinfo {note}
  {arXiv: 1810.10523}\BibitemShut {NoStop}%
\bibitem [{\citenamefont {Lin}\ \emph {et~al.}(2021)\citenamefont {Lin},
  \citenamefont {Wallman}, \citenamefont {Hincks},\ and\ \citenamefont
  {Laflamme}}]{lin21}%
  \BibitemOpen
  \bibfield  {author} {\bibinfo {author} {\bibfnamefont {Junan}\ \bibnamefont
  {Lin}}, \bibinfo {author} {\bibfnamefont {Joel~J.}\ \bibnamefont {Wallman}},
  \bibinfo {author} {\bibfnamefont {Ian}\ \bibnamefont {Hincks}}, \ and\
  \bibinfo {author} {\bibfnamefont {Raymond}\ \bibnamefont {Laflamme}},\
  }\bibfield  {title} {\enquote {\bibinfo {title} {Independent state and
  measurement characterization for quantum computers},}\ }\href {\doibase
  10.1103/PhysRevResearch.3.033285} {\bibfield  {journal} {\bibinfo  {journal}
  {Physical Review Research}\ }\textbf {\bibinfo {volume} {3}},\ \bibinfo
  {pages} {033285} (\bibinfo {year} {2021})},\ \bibinfo {note} {publisher:
  American Physical Society}\BibitemShut {NoStop}%
\bibitem [{\citenamefont {Werninghaus}\ \emph {et~al.}(2021)\citenamefont
  {Werninghaus}, \citenamefont {Egger},\ and\ \citenamefont
  {Filipp}}]{werninghaus21}%
  \BibitemOpen
  \bibfield  {author} {\bibinfo {author} {\bibfnamefont {M.}~\bibnamefont
  {Werninghaus}}, \bibinfo {author} {\bibfnamefont {D.J.}\ \bibnamefont
  {Egger}}, \ and\ \bibinfo {author} {\bibfnamefont {S.}~\bibnamefont
  {Filipp}},\ }\bibfield  {title} {\enquote {\bibinfo {title} {High-{Speed}
  {Calibration} and {Characterization} of {Superconducting} {Quantum}
  {Processors} without {Qubit} {Reset}},}\ }\href {\doibase
  10.1103/PRXQuantum.2.020324} {\bibfield  {journal} {\bibinfo  {journal} {PRX
  Quantum}\ }\textbf {\bibinfo {volume} {2}},\ \bibinfo {pages} {020324}
  (\bibinfo {year} {2021})},\ \bibinfo {note} {publisher: American Physical
  Society}\BibitemShut {NoStop}%
\bibitem [{\citenamefont {Milz}\ \emph {et~al.}(2017)\citenamefont {Milz},
  \citenamefont {Pollock},\ and\ \citenamefont {Modi}}]{milz17}%
  \BibitemOpen
  \bibfield  {author} {\bibinfo {author} {\bibfnamefont {Simon}\ \bibnamefont
  {Milz}}, \bibinfo {author} {\bibfnamefont {Felix~A.}\ \bibnamefont
  {Pollock}}, \ and\ \bibinfo {author} {\bibfnamefont {Kavan}\ \bibnamefont
  {Modi}},\ }\bibfield  {title} {\enquote {\bibinfo {title} {An introduction to
  operational quantum dynamics},}\ }\href {\doibase 10.1142/s1230161217400169}
  {\bibfield  {journal} {\bibinfo  {journal} {Open Systems \& Information
  Dynamics}\ }\textbf {\bibinfo {volume} {24}},\ \bibinfo {pages} {1740016}
  (\bibinfo {year} {2017})}\BibitemShut {NoStop}%
\bibitem [{\citenamefont {Breuer}\ \emph {et~al.}(2002)\citenamefont {Breuer},
  \citenamefont {Petruccione} \emph {et~al.}}]{breuer02}%
  \BibitemOpen
  \bibfield  {author} {\bibinfo {author} {\bibfnamefont {Heinz-Peter}\
  \bibnamefont {Breuer}}, \bibinfo {author} {\bibfnamefont {Francesco}\
  \bibnamefont {Petruccione}},  \emph {et~al.},\ }\href@noop {} {\emph
  {\bibinfo {title} {The theory of open quantum systems}}}\ (\bibinfo
  {publisher} {Oxford University Press on Demand},\ \bibinfo {year}
  {2002})\BibitemShut {NoStop}%
\bibitem [{\citenamefont {Bengtsson}\ and\ \citenamefont
  {{\.Z}yczkowski}(2017)}]{bengtsson17}%
  \BibitemOpen
  \bibfield  {author} {\bibinfo {author} {\bibfnamefont {Ingemar}\ \bibnamefont
  {Bengtsson}}\ and\ \bibinfo {author} {\bibfnamefont {Karol}\ \bibnamefont
  {{\.Z}yczkowski}},\ }\href@noop {} {\emph {\bibinfo {title} {Geometry of
  quantum states: an introduction to quantum entanglement}}}\ (\bibinfo
  {publisher} {Cambridge university press},\ \bibinfo {year}
  {2017})\BibitemShut {NoStop}%
\bibitem [{\citenamefont {Farhi}\ \emph {et~al.}(2014)\citenamefont {Farhi},
  \citenamefont {Goldstone},\ and\ \citenamefont {Gutmann}}]{farhi14}%
  \BibitemOpen
  \bibfield  {author} {\bibinfo {author} {\bibfnamefont {Edward}\ \bibnamefont
  {Farhi}}, \bibinfo {author} {\bibfnamefont {Jeffrey}\ \bibnamefont
  {Goldstone}}, \ and\ \bibinfo {author} {\bibfnamefont {Sam}\ \bibnamefont
  {Gutmann}},\ }\bibfield  {title} {\enquote {\bibinfo {title} {A {Quantum}
  {Approximate} {Optimization} {Algorithm}},}\ }\href
  {http://arxiv.org/abs/1411.4028} {\bibfield  {journal} {\bibinfo  {journal}
  {arXiv:1411.4028 [quant-ph]}\ } (\bibinfo {year} {2014})},\ \bibinfo {note}
  {arXiv: 1411.4028}\BibitemShut {NoStop}%
\bibitem [{\citenamefont {Hogg}(2000)}]{hogg00}%
  \BibitemOpen
  \bibfield  {author} {\bibinfo {author} {\bibfnamefont {Tad}\ \bibnamefont
  {Hogg}},\ }\bibfield  {title} {\enquote {\bibinfo {title} {Quantum search
  heuristics},}\ }\href {\doibase 10.1103/PhysRevA.61.052311} {\bibfield
  {journal} {\bibinfo  {journal} {Physical Review A}\ }\textbf {\bibinfo
  {volume} {61}},\ \bibinfo {pages} {052311} (\bibinfo {year} {2000})},\
  \bibinfo {note} {publisher: American Physical Society}\BibitemShut {NoStop}%
\bibitem [{\citenamefont {Hadfield}\ \emph {et~al.}(2019)\citenamefont
  {Hadfield}, \citenamefont {Wang}, \citenamefont {O’Gorman}, \citenamefont
  {Rieffel}, \citenamefont {Venturelli},\ and\ \citenamefont
  {Biswas}}]{hadfield19}%
  \BibitemOpen
  \bibfield  {author} {\bibinfo {author} {\bibfnamefont {Stuart}\ \bibnamefont
  {Hadfield}}, \bibinfo {author} {\bibfnamefont {Zhihui}\ \bibnamefont {Wang}},
  \bibinfo {author} {\bibfnamefont {Bryan}\ \bibnamefont {O’Gorman}},
  \bibinfo {author} {\bibfnamefont {Eleanor~G.}\ \bibnamefont {Rieffel}},
  \bibinfo {author} {\bibfnamefont {Davide}\ \bibnamefont {Venturelli}}, \ and\
  \bibinfo {author} {\bibfnamefont {Rupak}\ \bibnamefont {Biswas}},\ }\bibfield
   {title} {\enquote {\bibinfo {title} {From the {Quantum} {Approximate}
  {Optimization} {Algorithm} to a {Quantum} {Alternating} {Operator}
  {Ansatz}},}\ }\href {\doibase 10.3390/a12020034} {\bibfield  {journal}
  {\bibinfo  {journal} {Algorithms}\ }\textbf {\bibinfo {volume} {12}},\
  \bibinfo {pages} {34} (\bibinfo {year} {2019})},\ \bibinfo {note} {number: 2
  Publisher: Multidisciplinary Digital Publishing Institute}\BibitemShut
  {NoStop}%
\bibitem [{\citenamefont {Marshall}\ \emph {et~al.}(2020)\citenamefont
  {Marshall}, \citenamefont {Wudarski}, \citenamefont {Hadfield},\ and\
  \citenamefont {Hogg}}]{marshall20}%
  \BibitemOpen
  \bibfield  {author} {\bibinfo {author} {\bibfnamefont {Jeffrey}\ \bibnamefont
  {Marshall}}, \bibinfo {author} {\bibfnamefont {Filip}\ \bibnamefont
  {Wudarski}}, \bibinfo {author} {\bibfnamefont {Stuart}\ \bibnamefont
  {Hadfield}}, \ and\ \bibinfo {author} {\bibfnamefont {Tad}\ \bibnamefont
  {Hogg}},\ }\bibfield  {title} {\enquote {\bibinfo {title} {Characterizing
  local noise in {QAOA} circuits},}\ }\href {\doibase 10.1088/2633-1357/abb0d7}
  {\ \textbf {\bibinfo {volume} {1}},\ \bibinfo {pages} {025208} (\bibinfo
  {year} {2020})},\ \bibinfo {note} {publisher: IOP Publishing}\BibitemShut
  {NoStop}%
\bibitem [{\citenamefont {Xue}\ \emph {et~al.}(2021)\citenamefont {Xue},
  \citenamefont {Chen}, \citenamefont {Wu},\ and\ \citenamefont {Guo}}]{xue21}%
  \BibitemOpen
  \bibfield  {author} {\bibinfo {author} {\bibfnamefont {Cheng}\ \bibnamefont
  {Xue}}, \bibinfo {author} {\bibfnamefont {Zhao-Yun}\ \bibnamefont {Chen}},
  \bibinfo {author} {\bibfnamefont {Yu-Chun}\ \bibnamefont {Wu}}, \ and\
  \bibinfo {author} {\bibfnamefont {Guo-Ping}\ \bibnamefont {Guo}},\ }\bibfield
   {title} {\enquote {\bibinfo {title} {Effects of {Quantum} {Noise} on
  {Quantum} {Approximate} {Optimization} {Algorithm}},}\ }\href {\doibase
  10.1088/0256-307X/38/3/030302} {\ \textbf {\bibinfo {volume} {38}},\ \bibinfo
  {pages} {030302} (\bibinfo {year} {2021})},\ \bibinfo {note} {publisher: IOP
  Publishing}\BibitemShut {NoStop}%
\bibitem [{\citenamefont {Wang}\ \emph {et~al.}(2021)\citenamefont {Wang},
  \citenamefont {Fontana}, \citenamefont {Cerezo}, \citenamefont {Sharma},
  \citenamefont {Sone}, \citenamefont {Cincio},\ and\ \citenamefont
  {Coles}}]{wang2021noiseinduced}%
  \BibitemOpen
  \bibfield  {author} {\bibinfo {author} {\bibfnamefont {Samson}\ \bibnamefont
  {Wang}}, \bibinfo {author} {\bibfnamefont {Enrico}\ \bibnamefont {Fontana}},
  \bibinfo {author} {\bibfnamefont {M.}~\bibnamefont {Cerezo}}, \bibinfo
  {author} {\bibfnamefont {Kunal}\ \bibnamefont {Sharma}}, \bibinfo {author}
  {\bibfnamefont {Akira}\ \bibnamefont {Sone}}, \bibinfo {author}
  {\bibfnamefont {Lukasz}\ \bibnamefont {Cincio}}, \ and\ \bibinfo {author}
  {\bibfnamefont {Patrick~J.}\ \bibnamefont {Coles}},\ }\href@noop {} {\enquote
  {\bibinfo {title} {Noise-induced barren plateaus in variational quantum
  algorithms},}\ } (\bibinfo {year} {2021}),\ \Eprint
  {http://arxiv.org/abs/2007.14384} {arXiv:2007.14384 [quant-ph]} \BibitemShut
  {NoStop}%
\bibitem [{\citenamefont {Streif}\ \emph {et~al.}(2021)\citenamefont {Streif},
  \citenamefont {Leib}, \citenamefont {Wudarski}, \citenamefont {Rieffel},\
  and\ \citenamefont {Wang}}]{streif2021}%
  \BibitemOpen
  \bibfield  {author} {\bibinfo {author} {\bibfnamefont {Michael}\ \bibnamefont
  {Streif}}, \bibinfo {author} {\bibfnamefont {Martin}\ \bibnamefont {Leib}},
  \bibinfo {author} {\bibfnamefont {Filip}\ \bibnamefont {Wudarski}}, \bibinfo
  {author} {\bibfnamefont {Eleanor}\ \bibnamefont {Rieffel}}, \ and\ \bibinfo
  {author} {\bibfnamefont {Zhihui}\ \bibnamefont {Wang}},\ }\bibfield  {title}
  {\enquote {\bibinfo {title} {Quantum algorithms with local particle-number
  conservation: Noise effects and error correction},}\ }\href {\doibase
  10.1103/physreva.103.042412} {\bibfield  {journal} {\bibinfo  {journal}
  {Physical Review A}\ }\textbf {\bibinfo {volume} {103}} (\bibinfo {year}
  {2021}),\ 10.1103/physreva.103.042412}\BibitemShut {NoStop}%
\bibitem [{\citenamefont {Streif}\ and\ \citenamefont {Leib}(2020)}]{streif20}%
  \BibitemOpen
  \bibfield  {author} {\bibinfo {author} {\bibfnamefont {Michael}\ \bibnamefont
  {Streif}}\ and\ \bibinfo {author} {\bibfnamefont {Martin}\ \bibnamefont
  {Leib}},\ }\bibfield  {title} {\enquote {\bibinfo {title} {Training the
  quantum approximate optimization algorithm without access to a quantum
  processing unit},}\ }\href {\doibase 10.1088/2058-9565/ab8c2b} {\ \textbf
  {\bibinfo {volume} {5}},\ \bibinfo {pages} {034008} (\bibinfo {year}
  {2020})},\ \bibinfo {note} {publisher: IOP Publishing}\BibitemShut {NoStop}%
\bibitem [{\citenamefont {Zhou}\ \emph {et~al.}(2020)\citenamefont {Zhou},
  \citenamefont {Wang}, \citenamefont {Choi}, \citenamefont {Pichler},\ and\
  \citenamefont {Lukin}}]{zhou20}%
  \BibitemOpen
  \bibfield  {author} {\bibinfo {author} {\bibfnamefont {Leo}\ \bibnamefont
  {Zhou}}, \bibinfo {author} {\bibfnamefont {Sheng-Tao}\ \bibnamefont {Wang}},
  \bibinfo {author} {\bibfnamefont {Soonwon}\ \bibnamefont {Choi}}, \bibinfo
  {author} {\bibfnamefont {Hannes}\ \bibnamefont {Pichler}}, \ and\ \bibinfo
  {author} {\bibfnamefont {Mikhail~D.}\ \bibnamefont {Lukin}},\ }\bibfield
  {title} {\enquote {\bibinfo {title} {Quantum {Approximate} {Optimization}
  {Algorithm}: {Performance}, {Mechanism}, and {Implementation} on
  {Near}-{Term} {Devices}},}\ }\href {\doibase 10.1103/PhysRevX.10.021067}
  {\bibfield  {journal} {\bibinfo  {journal} {Physical Review X}\ }\textbf
  {\bibinfo {volume} {10}},\ \bibinfo {pages} {021067} (\bibinfo {year}
  {2020})},\ \bibinfo {note} {publisher: American Physical Society}\BibitemShut
  {NoStop}%
\bibitem [{\citenamefont {Shaydulin}\ \emph {et~al.}(2019)\citenamefont
  {Shaydulin}, \citenamefont {Safro},\ and\ \citenamefont
  {Larson}}]{shaydulin19}%
  \BibitemOpen
  \bibfield  {author} {\bibinfo {author} {\bibfnamefont {Ruslan}\ \bibnamefont
  {Shaydulin}}, \bibinfo {author} {\bibfnamefont {Ilya}\ \bibnamefont {Safro}},
  \ and\ \bibinfo {author} {\bibfnamefont {Jeffrey}\ \bibnamefont {Larson}},\
  }\bibfield  {title} {\enquote {\bibinfo {title} {Multistart {Methods} for
  {Quantum} {Approximate} optimization},}\ }in\ \href {\doibase
  10.1109/HPEC.2019.8916288} {\emph {\bibinfo {booktitle} {2019 {IEEE} {High}
  {Performance} {Extreme} {Computing} {Conference} ({HPEC})}}}\ (\bibinfo
  {year} {2019})\ pp.\ \bibinfo {pages} {1--8},\ \bibinfo {note} {iSSN:
  2643-1971}\BibitemShut {NoStop}%
\bibitem [{\citenamefont {Brandao}\ \emph {et~al.}(2018)\citenamefont
  {Brandao}, \citenamefont {Broughton}, \citenamefont {Farhi}, \citenamefont
  {Gutmann},\ and\ \citenamefont {Neven}}]{brandao18}%
  \BibitemOpen
  \bibfield  {author} {\bibinfo {author} {\bibfnamefont {Fernando G. S.~L.}\
  \bibnamefont {Brandao}}, \bibinfo {author} {\bibfnamefont {Michael}\
  \bibnamefont {Broughton}}, \bibinfo {author} {\bibfnamefont {Edward}\
  \bibnamefont {Farhi}}, \bibinfo {author} {\bibfnamefont {Sam}\ \bibnamefont
  {Gutmann}}, \ and\ \bibinfo {author} {\bibfnamefont {Hartmut}\ \bibnamefont
  {Neven}},\ }\bibfield  {title} {\enquote {\bibinfo {title} {For {Fixed}
  {Control} {Parameters} the {Quantum} {Approximate} {Optimization}
  {Algorithm}'s {Objective} {Function} {Value} {Concentrates} for {Typical}
  {Instances}},}\ }\href {http://arxiv.org/abs/1812.04170} {\bibfield
  {journal} {\bibinfo  {journal} {arXiv:1812.04170 [quant-ph]}\ } (\bibinfo
  {year} {2018})},\ \bibinfo {note} {arXiv: 1812.04170}\BibitemShut {NoStop}%
\bibitem [{\citenamefont {Sack}\ and\ \citenamefont {Serbyn}(2021)}]{sack21}%
  \BibitemOpen
  \bibfield  {author} {\bibinfo {author} {\bibfnamefont {Stefan~H.}\
  \bibnamefont {Sack}}\ and\ \bibinfo {author} {\bibfnamefont {Maksym}\
  \bibnamefont {Serbyn}},\ }\bibfield  {title} {\enquote {\bibinfo {title}
  {Quantum annealing initialization of the quantum approximate optimization
  algorithm},}\ }\href {\doibase 10.22331/q-2021-07-01-491} {\bibfield
  {journal} {\bibinfo  {journal} {Quantum}\ }\textbf {\bibinfo {volume} {5}},\
  \bibinfo {pages} {491} (\bibinfo {year} {2021})},\ \bibinfo {note}
  {publisher: Verein zur Förderung des Open Access Publizierens in den
  Quantenwissenschaften}\BibitemShut {NoStop}%
\bibitem [{\citenamefont {Wang}\ \emph {et~al.}(2020)\citenamefont {Wang},
  \citenamefont {Rubin}, \citenamefont {Dominy},\ and\ \citenamefont
  {Rieffel}}]{wang20}%
  \BibitemOpen
  \bibfield  {author} {\bibinfo {author} {\bibfnamefont {Zhihui}\ \bibnamefont
  {Wang}}, \bibinfo {author} {\bibfnamefont {Nicholas~C.}\ \bibnamefont
  {Rubin}}, \bibinfo {author} {\bibfnamefont {Jason~M.}\ \bibnamefont
  {Dominy}}, \ and\ \bibinfo {author} {\bibfnamefont {Eleanor~G.}\ \bibnamefont
  {Rieffel}},\ }\bibfield  {title} {\enquote {\bibinfo {title} {{$XY$ mixers:
  Analytical and numerical results for the quantum alternating operator
  ansatz}},}\ }\href {\doibase 10.1103/PhysRevA.101.012320} {\bibfield
  {journal} {\bibinfo  {journal} {Physical Review A}\ }\textbf {\bibinfo
  {volume} {101}},\ \bibinfo {pages} {012320} (\bibinfo {year} {2020})},\
  \bibinfo {note} {publisher: American Physical Society}\BibitemShut {NoStop}%
\bibitem [{\citenamefont {Glover}\ \emph {et~al.}(2019)\citenamefont {Glover},
  \citenamefont {Kochenberger},\ and\ \citenamefont {Du}}]{glover19}%
  \BibitemOpen
  \bibfield  {author} {\bibinfo {author} {\bibfnamefont {Fred}\ \bibnamefont
  {Glover}}, \bibinfo {author} {\bibfnamefont {Gary}\ \bibnamefont
  {Kochenberger}}, \ and\ \bibinfo {author} {\bibfnamefont {Yu}~\bibnamefont
  {Du}},\ }\bibfield  {title} {\enquote {\bibinfo {title} {A {Tutorial} on
  {Formulating} and {Using} {QUBO} {Models}},}\ }\href
  {http://arxiv.org/abs/1811.11538} {\bibfield  {journal} {\bibinfo  {journal}
  {arXiv:1811.11538 [quant-ph]}\ } (\bibinfo {year} {2019})},\ \bibinfo {note}
  {arXiv: 1811.11538}\BibitemShut {NoStop}%
\bibitem [{\citenamefont {Shaydulin}\ \emph {et~al.}(2021)\citenamefont
  {Shaydulin}, \citenamefont {Hadfield}, \citenamefont {Hogg},\ and\
  \citenamefont {Safro}}]{shaydulin21}%
  \BibitemOpen
  \bibfield  {author} {\bibinfo {author} {\bibfnamefont {Ruslan}\ \bibnamefont
  {Shaydulin}}, \bibinfo {author} {\bibfnamefont {Stuart}\ \bibnamefont
  {Hadfield}}, \bibinfo {author} {\bibfnamefont {Tad}\ \bibnamefont {Hogg}}, \
  and\ \bibinfo {author} {\bibfnamefont {Ilya}\ \bibnamefont {Safro}},\
  }\bibfield  {title} {\enquote {\bibinfo {title} {Classical symmetries and the
  {Quantum} {Approximate} {Optimization} {Algorithm}},}\ }\href {\doibase
  10.1007/s11128-021-03298-4} {\bibfield  {journal} {\bibinfo  {journal}
  {Quantum Information Processing}\ }\textbf {\bibinfo {volume} {20}},\
  \bibinfo {pages} {359} (\bibinfo {year} {2021})},\ \bibinfo {note} {arXiv:
  2012.04713}\BibitemShut {NoStop}%
\bibitem [{Note1()}]{Note1}%
  \BibitemOpen
  \bibinfo {note} {Extension to parameter-free circuit is straightforward, and
  requires only altering some gates, e.g. $X\to Z$. However, this procedure
  would disturb the investigated algorithm, and could serve only as a
  characterization protocol.}\BibitemShut {Stop}%
\bibitem [{\citenamefont {Peng}\ \emph {et~al.}(2020)\citenamefont {Peng},
  \citenamefont {Harrow}, \citenamefont {Ozols},\ and\ \citenamefont
  {Wu}}]{peng20}%
  \BibitemOpen
  \bibfield  {author} {\bibinfo {author} {\bibfnamefont {Tianyi}\ \bibnamefont
  {Peng}}, \bibinfo {author} {\bibfnamefont {Aram~W.}\ \bibnamefont {Harrow}},
  \bibinfo {author} {\bibfnamefont {Maris}\ \bibnamefont {Ozols}}, \ and\
  \bibinfo {author} {\bibfnamefont {Xiaodi}\ \bibnamefont {Wu}},\ }\bibfield
  {title} {\enquote {\bibinfo {title} {Simulating {Large} {Quantum} {Circuits}
  on a {Small} {Quantum} {Computer}},}\ }\href {\doibase
  10.1103/PhysRevLett.125.150504} {\bibfield  {journal} {\bibinfo  {journal}
  {Physical Review Letters}\ }\textbf {\bibinfo {volume} {125}},\ \bibinfo
  {pages} {150504} (\bibinfo {year} {2020})},\ \bibinfo {note} {publisher:
  American Physical Society}\BibitemShut {NoStop}%
\bibitem [{Note2()}]{Note2}%
  \BibitemOpen
  \bibinfo {note} {Here by mildly we mean, that noise is static, and parameter
  (e.g. angle) independent to the leading order.}\BibitemShut {Stop}%
\bibitem [{\citenamefont {Ash-Saki}\ \emph {et~al.}(2020)\citenamefont
  {Ash-Saki}, \citenamefont {Alam},\ and\ \citenamefont {Ghosh}}]{ash20}%
  \BibitemOpen
  \bibfield  {author} {\bibinfo {author} {\bibfnamefont {Abdullah}\
  \bibnamefont {Ash-Saki}}, \bibinfo {author} {\bibfnamefont {Mahabubul}\
  \bibnamefont {Alam}}, \ and\ \bibinfo {author} {\bibfnamefont {Swaroop}\
  \bibnamefont {Ghosh}},\ }\bibfield  {title} {\enquote {\bibinfo {title}
  {Experimental {Characterization}, {Modeling}, and {Analysis} of {Crosstalk}
  in a {Quantum} {Computer}},}\ }\href {\doibase 10.1109/TQE.2020.3023338}
  {\bibfield  {journal} {\bibinfo  {journal} {IEEE Transactions on Quantum
  Engineering}\ }\textbf {\bibinfo {volume} {1}},\ \bibinfo {pages} {1--6}
  (\bibinfo {year} {2020})},\ \bibinfo {note} {conference Name: IEEE
  Transactions on Quantum Engineering}\BibitemShut {NoStop}%
\bibitem [{\citenamefont {Sarovar}\ \emph {et~al.}(2020)\citenamefont
  {Sarovar}, \citenamefont {Proctor}, \citenamefont {Rudinger}, \citenamefont
  {Young}, \citenamefont {Nielsen},\ and\ \citenamefont
  {Blume-Kohout}}]{sarovar20}%
  \BibitemOpen
  \bibfield  {author} {\bibinfo {author} {\bibfnamefont {Mohan}\ \bibnamefont
  {Sarovar}}, \bibinfo {author} {\bibfnamefont {Timothy}\ \bibnamefont
  {Proctor}}, \bibinfo {author} {\bibfnamefont {Kenneth}\ \bibnamefont
  {Rudinger}}, \bibinfo {author} {\bibfnamefont {Kevin}\ \bibnamefont {Young}},
  \bibinfo {author} {\bibfnamefont {Erik}\ \bibnamefont {Nielsen}}, \ and\
  \bibinfo {author} {\bibfnamefont {Robin}\ \bibnamefont {Blume-Kohout}},\
  }\bibfield  {title} {\enquote {\bibinfo {title} {Detecting crosstalk errors
  in quantum information processors},}\ }\href {\doibase
  10.22331/q-2020-09-11-321} {\bibfield  {journal} {\bibinfo  {journal}
  {Quantum}\ }\textbf {\bibinfo {volume} {4}},\ \bibinfo {pages} {321}
  (\bibinfo {year} {2020})},\ \bibinfo {note} {publisher: Verein zur Förderung
  des Open Access Publizierens in den Quantenwissenschaften}\BibitemShut
  {NoStop}%
\bibitem [{\citenamefont {Greenberger}\ \emph {et~al.}(2007)\citenamefont
  {Greenberger}, \citenamefont {Horne},\ and\ \citenamefont
  {Zeilinger}}]{greenberger2007going}%
  \BibitemOpen
  \bibfield  {author} {\bibinfo {author} {\bibfnamefont {Daniel~M.}\
  \bibnamefont {Greenberger}}, \bibinfo {author} {\bibfnamefont {Michael~A.}\
  \bibnamefont {Horne}}, \ and\ \bibinfo {author} {\bibfnamefont {Anton}\
  \bibnamefont {Zeilinger}},\ }\href@noop {} {\enquote {\bibinfo {title} {Going
  beyond bell's theorem},}\ } (\bibinfo {year} {2007}),\ \Eprint
  {http://arxiv.org/abs/0712.0921} {arXiv:0712.0921 [quant-ph]} \BibitemShut
  {NoStop}%
\bibitem [{\citenamefont {Życzkowski}\ and\ \citenamefont
  {Sommers}(2005)}]{zyczkowski05}%
  \BibitemOpen
  \bibfield  {author} {\bibinfo {author} {\bibfnamefont {Karol}\ \bibnamefont
  {Życzkowski}}\ and\ \bibinfo {author} {\bibfnamefont {Hans-Jürgen}\
  \bibnamefont {Sommers}},\ }\bibfield  {title} {\enquote {\bibinfo {title}
  {Average fidelity between random quantum states},}\ }\href {\doibase
  10.1103/physreva.71.032313} {\bibfield  {journal} {\bibinfo  {journal}
  {Physical Review A}\ }\textbf {\bibinfo {volume} {71}} (\bibinfo {year}
  {2005}),\ 10.1103/physreva.71.032313}\BibitemShut {NoStop}%
\bibitem [{\citenamefont {Flammia}\ and\ \citenamefont
  {Wallman}(2020)}]{2020FlammiaEfficient}%
  \BibitemOpen
  \bibfield  {author} {\bibinfo {author} {\bibfnamefont {Steven~T.}\
  \bibnamefont {Flammia}}\ and\ \bibinfo {author} {\bibfnamefont {Joel~J.}\
  \bibnamefont {Wallman}},\ }\bibfield  {title} {\enquote {\bibinfo {title}
  {Efficient estimation of pauli channels},}\ }\href {\doibase 10.1145/3408039}
  {\bibfield  {journal} {\bibinfo  {journal} {ACM Transactions on Quantum
  Computing}\ }\textbf {\bibinfo {volume} {1}},\ \bibinfo {pages} {1–32}
  (\bibinfo {year} {2020})}\BibitemShut {NoStop}%
\bibitem [{\citenamefont {Johansson}\ \emph {et~al.}(2012)\citenamefont
  {Johansson}, \citenamefont {Nation},\ and\ \citenamefont
  {Nori}}]{johansson12}%
  \BibitemOpen
  \bibfield  {author} {\bibinfo {author} {\bibfnamefont {J.~R.}\ \bibnamefont
  {Johansson}}, \bibinfo {author} {\bibfnamefont {P.~D.}\ \bibnamefont
  {Nation}}, \ and\ \bibinfo {author} {\bibfnamefont {Franco}\ \bibnamefont
  {Nori}},\ }\bibfield  {title} {\enquote {\bibinfo {title} {{QuTiP}: {An}
  open-source {Python} framework for the dynamics of open quantum systems},}\
  }\href {\doibase 10.1016/j.cpc.2012.02.021} {\bibfield  {journal} {\bibinfo
  {journal} {Computer Physics Communications}\ }\textbf {\bibinfo {volume}
  {183}},\ \bibinfo {pages} {1760--1772} (\bibinfo {year} {2012})},\ \bibinfo
  {note} {arXiv: 1110.0573}\BibitemShut {NoStop}%
\bibitem [{Note3()}]{Note3}%
  \BibitemOpen
  \bibinfo {note} {The marginal approximation for Pauli channels is also done
  on the level of matrices, and not on the probability vectors.}\BibitemShut
  {Stop}%
\bibitem [{\citenamefont {Hurwitz}(1897)}]{Hurwitz1897}%
  \BibitemOpen
  \bibfield  {author} {\bibinfo {author} {\bibfnamefont {A.}~\bibnamefont
  {Hurwitz}},\ }\bibfield  {title} {\enquote {\bibinfo {title} {über die
  erzeugung der invarianten durch integration},}\ }\href
  {http://eudml.org/doc/58378} {\bibfield  {journal} {\bibinfo  {journal}
  {Nachrichten von der Gesellschaft der Wissenschaften zu Göttingen,
  Mathematisch-Physikalische Klasse}\ }\textbf {\bibinfo {volume} {1897}},\
  \bibinfo {pages} {71--2} (\bibinfo {year} {1897})}\BibitemShut {NoStop}%
\bibitem [{\citenamefont {Zyczkowski}\ and\ \citenamefont
  {Sommers}(2001)}]{Zyczkowski2001}%
  \BibitemOpen
  \bibfield  {author} {\bibinfo {author} {\bibfnamefont {Karol}\ \bibnamefont
  {Zyczkowski}}\ and\ \bibinfo {author} {\bibfnamefont {Hans-Jürgen}\
  \bibnamefont {Sommers}},\ }\bibfield  {title} {\enquote {\bibinfo {title}
  {Induced measures in the space of mixed quantum states},}\ }\href {\doibase
  10.1088/0305-4470/34/35/335} {\bibfield  {journal} {\bibinfo  {journal}
  {Journal of Physics A: Mathematical and General}\ }\textbf {\bibinfo {volume}
  {34}},\ \bibinfo {pages} {7111–7125} (\bibinfo {year} {2001})}\BibitemShut
  {NoStop}%
\bibitem [{\citenamefont {Kraft}\ \emph {et~al.}(1988)\citenamefont {Kraft}
  \emph {et~al.}}]{kraft88}%
  \BibitemOpen
  \bibfield  {author} {\bibinfo {author} {\bibfnamefont {Dieter}\ \bibnamefont
  {Kraft}} \emph {et~al.},\ }\bibfield  {title} {\enquote {\bibinfo {title} {A
  software package for sequential quadratic programming},}\ }\href@noop {} {\
  (\bibinfo {year} {1988})}\BibitemShut {NoStop}%
\bibitem [{\citenamefont {Mandrà}\ \emph {et~al.}(2021)\citenamefont
  {Mandrà}, \citenamefont {Marshall}, \citenamefont {Rieffel},\ and\
  \citenamefont {Biswas}}]{mandra21}%
  \BibitemOpen
  \bibfield  {author} {\bibinfo {author} {\bibfnamefont {Salvatore}\
  \bibnamefont {Mandrà}}, \bibinfo {author} {\bibfnamefont {Jeffrey}\
  \bibnamefont {Marshall}}, \bibinfo {author} {\bibfnamefont {Eleanor~G.}\
  \bibnamefont {Rieffel}}, \ and\ \bibinfo {author} {\bibfnamefont {Rupak}\
  \bibnamefont {Biswas}},\ }\href@noop {} {\enquote {\bibinfo {title}
  {{HybridQ: A Hybrid Simulator for Quantum Circuits}},}\ } (\bibinfo {year}
  {2021}),\ \Eprint {http://arxiv.org/abs/2111.06868} {arXiv:2111.06868
  [quant-ph]} \BibitemShut {NoStop}%
\bibitem [{\citenamefont {Abrams}\ \emph {et~al.}(2020)\citenamefont {Abrams},
  \citenamefont {Didier}, \citenamefont {Johnson}, \citenamefont {Silva},\ and\
  \citenamefont {Ryan}}]{abrams20}%
  \BibitemOpen
  \bibfield  {author} {\bibinfo {author} {\bibfnamefont {Deanna~M.}\
  \bibnamefont {Abrams}}, \bibinfo {author} {\bibfnamefont {Nicolas}\
  \bibnamefont {Didier}}, \bibinfo {author} {\bibfnamefont {Blake~R.}\
  \bibnamefont {Johnson}}, \bibinfo {author} {\bibfnamefont {Marcus P.~da}\
  \bibnamefont {Silva}}, \ and\ \bibinfo {author} {\bibfnamefont {Colm~A.}\
  \bibnamefont {Ryan}},\ }\bibfield  {title} {\enquote {\bibinfo {title}
  {Implementation of {XY} entangling gates with a single calibrated pulse},}\
  }\href {\doibase 10.1038/s41928-020-00498-1} {\bibfield  {journal} {\bibinfo
  {journal} {Nature Electronics}\ }\textbf {\bibinfo {volume} {3}},\ \bibinfo
  {pages} {744--750} (\bibinfo {year} {2020})}\BibitemShut {NoStop}%
\bibitem [{rig()}]{rigetti}%
  \BibitemOpen
  \href {https://www.rigetti.com/} {\enquote {\bibinfo {title} {Rigetti quantum
  computing},}\ }\bibinfo {note} {Https://www.rigetti.com/}\BibitemShut
  {NoStop}%
\end{thebibliography}%


%

\onecolumngrid

\appendix
\section{More error maps}\label{app:error_maps}

In this section we derive the effects of various other error channels on the pauli $Z$ and $ZZ$ terms found in QUBO problems, in the same style as section \ref{noisysinglelayeredqaoa}. We first present the following proof to aid in the analysis:

\newtheorem{thm}{Theorem}
\begin{thm}\label{thm:odd_z_y}
For $\mathbb{Z}_2$ symmetric states, the expectation value of a Pauli string $\mathcal{P}$ is $0$ if the number of Pauli Z's plus Pauli Y's in $\mathcal{P}$ is odd.
\end{thm}

\begin{proof}
We assume an operator $\mathcal{O}$ on $N$ qubits that is of the form $\prod_{i=1}^N \sigma_{p_i}$ where $p_i$ represents the Pauli that acts on qubit $i$: either X,Y,Z, or I. We can then define $S_X,S_Y,S_Z$ to be the set of qubits in which our operator is X,Y,and Z. We then look at the expectation value of the most general operator of this form
\begin{equation}
    \bra{\Psi}\mathcal{O}\ket{\Psi} = 
    \bra{\Psi} \prod_{i\in S_Y}Z_i \prod_{j\in S_Y}Y_j \prod_{k\in S_X}X_k \ket{\Psi}.
\end{equation}
Then noting that Y = -iZX and write
\begin{equation}
    -i\bra{\Psi} \prod_{i\in S_Z}Z_i \prod_{j\in S_Y}Z_jX_j \prod_{k\in S_X}X_k \ket{\Psi} =
    -i\bra{\Psi} \prod_{i\in S_Z \cup S_Y}Z_i \prod_{k\in S_X \cup S_Y}X_k \ket{\Psi}.
\end{equation}
We can then expand out $\ket{\Psi}$ in terms of bitstrings $l$
\begin{equation}
    -i\bra{\Psi} \prod_{i\in S_Z \cup S_Y}Z_i \prod_{k\in S_X \cup S_Y}X_k 
    \sum_l c_l (\ket{l}+\ket{\overline{l}}).
\end{equation}
We can then define $\ket{l^*} = \prod_{k\in S_X \cup S_Y}X_k \ket{l}$ and write
\begin{equation}
    -i\bra{\Psi} \prod_{i\in S_Z \cup S_Y}Z_i \sum_l c_l (\ket{l^*}+\ket{\overline{l^*}})
\end{equation}
\begin{equation}
     -i \sum_l c_l (\bra{\Psi} \prod_{i\in S_Z \cup S_Y}Z_i \ket{l^*} +
    (\bra{\Psi} \prod_{i\in S_Z \cup S_Y}Z_i \ket{\overline{l^*}})
\end{equation}.
Then we can note following two properties. First, since $\Phi$ is $\mathbb{Z}_2$ symmetric, $\braket{\Phi|l^*} = \braket{\Phi|\overline{l^*}}$ by definition. Also, for general state $\ket{l}$, we note that $Z_i\ket{\overline{l}}=\ket{\overline{l}}\bra{l}(-Z_i)\ket{l}$ for $i$ on any qubit. 
Using the second property we have
\begin{equation}
    -i \sum_l c_l (\bra{\Psi} \prod_{i\in S_Z \cup S_Y}Z_i \ket{l^*} +
    (\bra{\Psi} \prod_{i\in S_Z \cup S_Y}(-Z_i) \ket{\overline{l^*}}\bra{l^*}(-Z)\ket{l^*}).
\end{equation}
where $\ket{\overline{x}}$ represent the inverse of $\ket{x}$, obtained by flipping all qubits in the state.
Rewriting and then using the first property we see
\begin{gather}
    -i \sum_l c_l (\bra{\Psi} \prod_{i\in S_Z \cup S_Y}Z_i \ket{l^*} + (-1)^{S_Y+S_Z}
    \bra{\Psi} \prod_{i\in S_Z \cup S_Y} \ket{\overline{l^*}}\bra{l^*}(-Z_i)\ket{l^*}) \\
    =-i \sum_l c_l (\bra{\Psi} \prod_{i\in S_Z \cup S_Y}Z_i \ket{l^*} + (-1)^{|S_Y|+|S_Z|}
    \bra{\Psi} \prod_{i\in S_Z \cup S_Y} \ket{\overline{l^*}}\bra{l^*}Z_i\ket{l^*}) \\
    =-i \sum_l c_l (\bra{\Psi} \prod_{i\in S_Z \cup S_Y}Z_i \ket{l^*} + (-1)^{|S_Y|+|S_Z|}
    \bra{\Psi} \prod_{i\in S_Z \cup S_Y} \ket{\overline{l^*}}\bra{l^*}Z_i\ket{l^*}) \\
    =-i \sum_l c_l (\bra{\Psi} \prod_{i\in S_Z \cup S_Y}Z_i \ket{l^*} + (-1)^{|S_Y|+|S_Z|}
    \bra{\Psi} \prod_{i\in S_Z \cup S_Y} Z_i\ket{l^*}) \\
    =-i \sum_l c_l (1+(-1)^{|S_Y|+|S_Z|}) \bra{\Psi} \prod_{i\in S_Z \cup S_Y}Z_i \ket{l^*}.
\end{gather}
Where we use the first property again in the second to last step. Now we can clearly see that if there if $|S_Y| + |S_Z|$ is odd this inner product will vanish for all $l$
\end{proof}
We will reference this theorem in the following analyses

\subsection{Generic Single Qubit Channel}\label{generic_single_qubit_channel}
The action of the dual of a generic single qubit channel defined by $\chi$ in \ref{eq:chi_map} on $Z$ and $ZZ$ is as follows
\begin{align}
    &\ep^\#(Z) = p_Z Z + p_Y Y + p_X X + p_I I \label{eq:genericchannelz},\\
    &\ep^\#(ZZ) = p_Z^2 ZZ + p_Y^2 YY + p_X^2 XX + p_I^2 I  + \nonumber \\
    &+P_ZP_Y(ZY+YZ)+P_ZP_X(ZX+XZ) \nonumber \\ 
    &+P_ZP_I(ZI+IZ)+P_YP_X(YX+XY) \nonumber \\ 
    &+P_YP_I(YI+IY)+P_XP_I(XI+IX). \label{eq:genericchannelzz}
\end{align}
where $p_Z = p_0+p_3-p_1-p_2$, $p_Y = 2(t_{23}+v_{02})$, $p_X = 2(t_{13}-v_{02})$, and $p_I = 4t_{03}$.

Thus, the action on QAOA Hamiltonians of form given in Eq.~\eqref{eq:localQAOAop} is:

\begin{align}
    H' = 
    &\sum_i h_i (p_Z Z_i + p_Y Y_i + p_X X_i + P_I I) \\ \nonumber 
    &+\sum_{i<j} J_{ij} (p_Z^2 Z_iZ_j + p_Y^2 Y_iY_j + p_X^2 X_iX_j + p_I^2 I ) \nonumber \\
    &+\sum_{i,j} J_{i,j}(P_ZP_YZ_iY_j+P_ZP_XZ_iX_j+P_ZP_IZ_iI \nonumber \\
    &+P_YP_XY_iX_j+P_YP_IY_iI+P_XP_IX_iI).
\end{align}

From here various assumptions can be made.  If all h's=0, which is the case for MaxCut and strictly 2-local QUBO problems we can eliminate all terms with odd number of $Z$+$Y$ terms from Thm.\ref{thm:odd_z_y} in the appendix. We could also assume that all but $p_Z$ are small, meaning that the noise channel is relatively close to the identity, which is a condition that would be satisfied on quantum hardware with low levels of noise. This would allow us to eliminate all terms quadratic in $p_X$,$p_Y$,and $p_Z$. If we make these two assumptions we reduce to
\begin{equation}
    H' = p_Z^2H + p_Zp_Y\sum_{i,j}J_{ij}Z_iY_j.
\end{equation}
This corresponds to a simple rescaling of the Hamiltonian, plus an additional, nontrivial term, which is examined in the appendix Sec.~\ref{constant_mixing_overrotations}.

\subsection{Constant Mixing Overrotations}\label{constant_mixing_overrotations}
Take a very simple model, where the phase is applied correctly, but the mixer applied as $H_M = \sum_i \tilde{\beta}_i X_i$ where $\beta_i = \beta + \delta \beta_i$, i.e. a small over/under rotation in the x direction. This means the mixing unitary is of the form $U_M' U_M$ where $U_M = e^{-i\beta \sum_i X_i }$ and $U_M' = e^{-i \sum_i \delta \beta_i X_i}$.

This locally rotates the Hamiltonian, $H' = U_M'^\dag H U_M'$. Now use $e^{-i \theta X}Ze^{i\theta X} = \cos (2\theta) Z - \sin(2\theta) Y$.

Assuming the standard form $H=\sum_{i<j} J_{ij}Z_iZ_j$, 
\begin{equation}
\begin{split}
    & H' = \sum_{i<j} J_{ij}\cos(2\delta \beta_i)\cos(2\delta \beta_j) Z_i Z_j + \sum_{i<j} J_{ij}\cos (2\delta \beta_i) \sin(2\delta \beta_j)Z_iY_j \\
    & + \sum_{i<j} J_{ij}\sin (2\delta \beta_i \cos(2\delta \beta_j)Y_i Z_j + \sum_{i<j} J_{ij}\sin (2\delta \beta_i \sin(2\delta \beta_j)Y_iY_j. 
    \end{split}
\end{equation}

Let's average the fluctuations $\langle \cos 2\delta \beta_i \cos 2\delta \beta_j\rangle = \cos^2 (2\delta \beta)$ etc.

This gives the noise averaged Hamiltonian
\begin{equation}
    H' = \cos^2(2\delta \beta) H + \sin(4 \delta \beta)\sum_{i,j} \frac{J_{ij}}{2}Z_i Y_j + \sin^2(2\delta \beta) \sum_{i<j} J_{ij}Y_i Y_j.
\end{equation}
where in the middle sum, we now sum over all $i$ and all $j$.

We see that first, the spectrum is flattened by a factor of $\cos^2 (2\delta \beta)$, but second, the terms in $Y$ modify it in a non-trivial way. Let us look at a perturbation to order $\delta \beta$:
\begin{equation}
    H' = H + 2\delta \beta \sum_{i,j}J_{ij}Z_i Y_j + O(\delta \beta ^2).
\end{equation}

The first order correction to any energy level is
\begin{equation}
    \sum_{i,j} J_{ij} \langle E | Z_i Y_j |E\rangle = 0,
\end{equation}
using that $|E\rangle$ is just some $z$ bit string.

This suggests we must go to second order perturbation, looking at terms of the form
\begin{equation}
    \langle E_1| Y_j | E_2\rangle,
\end{equation}
where $E_2$ is a single bit-flip from $E_1$. Since the magnitude of the change depends on the energy difference $E_2 - E_1$, the correction depends strongly on the spectrum of the original problem.

Since the eigenstates of classical Hamiltonians are computational basis states, just like in the Hamming weight Hamiltonian, the shifted eigenstates are as before.
We can then again calculate:

\begin{gather}
        \braket{H'} = \sum_n  E_n \bigl\lvert \bra{m} (\frac{\ket{n} - i \epsilon \sum_k X_k \ket{n}}{\sqrt{1 + N \epsilon^2}}) \bigr\rvert^2 \label{eq:2_local_overrot}\\
        = \frac{1}{1+N\epsilon^2}(E_m + \epsilon^2 \sum_n E_n \lvert \bra{m} \sum_k X_k \ket{n}) \rvert ^2.
\end{gather}
We now note that the inner product is 1 if and only if $n$ is a bitflip away from $m$ and if $k$ is the index of the bit that is flipped. We can also calculate the energy difference between $E_m$ and $E_n$ in this case. Here we compute w.l.o.g the case where let $k$ be the vertex of highest index.
\begin{multline}
    E_m-E_n = 
    \Bigl(\sum_{\substack{i<j \\ j\neq k}}J_{ij}\braket{m|Z_iZ_j|m} + \sum_{i \neq k} J_{ik}\braket{m|Z_iZ_k|m}\Bigr) - \\ \Bigl(\sum_{\substack{i<j \\ j\neq k}}\braket{m|X_kZ_iZ_jX_k|m} + \sum_{i \neq k} J_{ik}\braket{m|X_kZ_iZ_kX_k|m}\Bigr).
\end{multline}
Then we use the fact that $X_kZ_kX_k=-Z_k$ and that $X_k$ commutes through $Z_iZ_j$. We can then cancel and add terms, giving us
\begin{equation}
    2 \sum_{i\neq k} J_{ik} \braket{m|Z_iZ_k|m}
\end{equation}
Then we can easily rearranged to see that $E_n$ = $E_m - 2 \sum_{i\neq k} J_{ik} \braket{m|Z_iZ_k|m}$. We can then plug this expression back in for $E_n$ in Eq.~\eqref{eq:2_local_overrot}:
\begin{gather}
    \frac{1}{1+N\epsilon^2}\Bigl(E_m + \epsilon^2 \bigl(\sum_k E_m - 2 \sum_{i\neq k} J_{ik} \braket{Z_iZ_k}\bigr)\Bigr)\\
    =\frac{1}{1+N\epsilon^2}\bigl(E_m + \epsilon^2 N E_m -4 \sum_{i<k} J_{ik}\braket{Z_iZ_k}\bigr)\\
    =E_m \bigl( \frac{1+(N-4)\epsilon^2}{1+N\epsilon^2} \bigr) \\
    \approx E_m (1-4\epsilon^2+4N\epsilon^4 + \mathcal{O}(\epsilon^6)).
\end{gather}

So for this case overrotations also just scales the old eigenvalues

\subsection{Non-constant Mixing Overrotations}\label{non_constant_mixing_overrotations}
We may consider that instead of having perfect angles, they demonstrate small stochastic fluctuations. We will exploit model based on von Mises distribution of angles (i.e. normal distribution on a circle) and we're looking for the maps
\begin{equation}
    \mathcal{E}_U^\alpha(\rho) = \int_0^{2\pi} \frac{e^{\kappa \cos(\varepsilon})}{2\pi I_0(\kappa)} U(\alpha + \varepsilon) \rho U^\dag(\alpha+\varepsilon)d \varepsilon,
\end{equation}
where $1/\kappa$ is variance, $I_0(\kappa)$ is modified Bessel function, and $U(\alpha + \varepsilon)$ is our mixer of phase operator set to angle $\alpha$ with fluctuation $\varepsilon$. 

The integral for mixer leads to the following map 
\begin{equation}
    \mathcal{E}_M^\beta(\rho) = \frac{1}{2}\Big(1+\frac{I_2(\kappa)}{I_0(\kappa)}\cos(2\beta)\Big) \rho + \frac{1}{2}\Big(1-\frac{I_2(\kappa)}{I_0(\kappa)}\cos(2\beta)\Big)X\rho X -\frac{i I_2(\kappa)}{I_0(\kappa)}\frac{\sin(2\beta)}{2}[X,\rho],
\end{equation}
while for a single gate of phase (assuming that we have cost Hamiltonian $H = \sum J_{ij} Z_iZ_j$), the single $ZZ(J_{ij} \gamma)$ gate (for $J_{ij}=\{+1,-1\}$) is given by
\begin{equation}
    \mathcal{E}_P^{\gamma,J}(\rho) = \frac{1}{2}\Big(1-\frac{I_2(\kappa)}{I_0(\kappa)}\Big) A_1(\gamma,J)\rho A_1^\dag(\gamma,J)+ \frac{1}{2}\Big(1+\frac{I_2(\kappa)}{I_0(\kappa)}\Big) A_2(\gamma,J) \rho A_2^\dag(\gamma,J),
\end{equation}
where
\begin{equation}
    A_1(\gamma,J) = \mathrm{Diag}\big[(1,-e^{2 i \gamma J},-e^{2 i \gamma J},1)\big] \quad, \quad A_2(\gamma,J) = \mathrm{Diag}\big[(1,e^{2 i \gamma J},e^{2 i \gamma J},1)\big].
\end{equation}

The mixer error map can also be looked at as a composition of a perfect rotation by $\beta$, then an application of $\mathcal{E}_M^\beta(\rho)$ with $\beta=0$. In this case we get a self-dual channel of form  
\begin{equation}
    \mathcal{E}_M^\beta(\rho)=\frac{1}{2}\Big(1+\frac{I_2(\kappa)}{I_0(\kappa)}\cos(2\beta)\Big) \rho + \frac{1}{2}\Big(1-\frac{I_2(\kappa)}{I_0(\kappa)}\cos(2\beta)\Big)X\rho X.
\end{equation}
This maps 
\begin{equation}
    Z \rightarrow \frac{I_2(\kappa)}{I_0(\kappa)} Z.
\end{equation}
And the analysis follows local depolarizing noise. We note that in the limit of high variance, $\kappa \rightarrow 0$, and $\frac{I_2(\kappa)}{I_0(\kappa)} \rightarrow 0$, so we have complete disorder in $Z$ ($Z \rightarrow 0)$. In the limit of 0 variance, $\kappa \rightarrow \infty$, and $\frac{I_2(\kappa)}{I_0(\kappa)} \rightarrow 1$, so there is no effect of the channel ($Z \rightarrow Z)$.

\subsection{Pauli Channel}\label{pauli_channel}
Interesting class of noisy channel is so-called Pauli channels, that have form
\begin{equation}
    \mathcal{E}(\rho) = \sum_{k=0}^3 p_k \sigma_k \rho \sigma_k,
\end{equation}
where $\sigma_0 = \mathbbm{1}$, and for $k=1,2,3$ we get Pauli X, Y, and Z, respectively. $p_k\ge0$ and $\sum_k p_k = 1$. One may interpret that a given noisy channel (e.g. bit-flip corresponding to $X$, happens with a respective probability). One can set $p_1=p_2=p_3 = (1-p)/4$ and $p_0 =(1+3p)/4$ and get depolarizing channel as in \eqref{eq:depolarizing}. 

The action of the dual of a Pauli channel on Z is:
\begin{equation}
    \mathcal{E}^\#(Z) = (p_0+p_3-(p_1+p_2))*Z,
\end{equation}
This then reduces to depolarizing noise where we set $p=p_0+p_3-(p_1+p_2)$.
The results are similar for other Pauli matrices, with the difference that for Pauli $k$ matrix (1=x, 2=y, 3=z), we will have $p_0 +p_k -(p_i + p_j)$, where i,j are the coefficients standing in front of remaining matrices.

\subsection{Phase damping}\label{phase_damping}
Phase damping has also two Krauses, with $A_1$ same as for the amplitude damping case analyzed in Sec \ref{amplitudedamping} and $A_2$ given in 
\begin{equation}
    A_2 = \left(\begin{array}{cc}0 & 0 \\0 & \sqrt{\gamma}\end{array}\right).
\end{equation}

In this case, the error maps Z to itself with no scaling or shifting, so the error channel has no effect on classical Hamiltonians.

\end{document}